\documentclass[11pt,a4paper,epsfig]{article}

\usepackage{amsmath}
\usepackage{amssymb}
\usepackage{algorithmic}
\usepackage{algorithm}
\usepackage{comment}

\usepackage{outliers}

\usepackage{color}
\usepackage{graphicx}

\usepackage{geometry}
\usepackage{times}
\usepackage{xspace}
\usepackage{ifpdf}
\usepackage{url,hyperref}
\usepackage{latexsym}
\usepackage{euscript}
\usepackage{color}
\usepackage{makeidx}
\usepackage{myref}

\hypersetup{                    
    colorlinks=true,                
    breaklinks=true,                
    urlcolor= blue,                 
    linkcolor= black,                
    citecolor= blue                
    }

\title{Gromov-Hausdorff Approximation of Metric Spaces with Linear Structure}

\author{Fr\'ed\'eric Chazal \thanks{INRIA Saclay - France - frederic.chazal@inria.fr}, Jian Sun \thanks{Tsinghua University - China}}

\begin{document}


\maketitle

\begin{abstract}In many real-world applications data come as discrete metric spaces sampled around 1-dimensional filamentary structures that can be seen as metric graphs. In this paper we address the metric reconstruction problem of such filamentary structures from data sampled around them. We prove that they can be approximated, with respect to the Gromov-Hausdorff distance by well-chosen Reeb graphs (and some of their variants) and we provide an efficient and easy to implement algorithm to compute such approximations in almost linear time. We illustrate the performances of our algorithm on a few synthetic and real data sets. 
\end{abstract}


\section{Introduction}
\paragraph{Motivation.} 
With the advance of sensor technology, computing power and Internet, massive
amounts of geometric data are being generated and collected in various areas of
science, engineering and business. As they are becoming widely available, there
is a real need to analyze and visualize these large scale geometric data to
extract useful information out of them. In many cases these data are not
embedded in Euclidean spaces and come as (finite) sets of points with pairwise
distances information, i.e. (discrete) metric spaces. A large amount of
research has been done on dimensionality reduction, manifold learning and
geometric inference for data embedded in, possibly high dimensional, Euclidean
spaces and assumed to be concentrated around low dimensional manifolds~\cite{belkin2003led,Lafon04diffusion,J.B.Tenenbaum:2000:0833c}. 
However, the assumption of data lying on a manifold may fail in
many applications. In addition, the strategy of representing data by points in
Euclidean space may introduce large metric distortions as the data may lie in
highly curved spaces, instead of in flat Euclidean space raising many
difficulties in the analysis of metric data. In the past decade, with the
development of topological methods in data analysis, new theories such as
topological persistence (see, for example, \cite{elz-tps-02, zc-cph-05, c-td-09, chazal2012structure}) and new tools such as the Mapper
algorithm \cite{smc-tmah-07} have given rise to new algorithms to extract and
visualize geometric and topological information from metric data without the
need of an embedding into an Euclidean space.  In this paper we focus on a
simple but important setting where the underlying geometric structure
approximating the data can be seen as a branching filamentary structure i.e.,
more precisely, as a {\em metric graph} which is a topological graph endowed
with a length assigned to each edge (see Section \ref{sec:metric-GH}). Such
structures appear naturally in various real-world data such as collections of
GPS traces collected by vehicles on a road network, earthquakes distributions
that concentrate around geological faults, distributions of galaxies in the
universe, networks of blood vessels in anatomy or hydrographic networks in
geography just to name a few of them. 
It is thus appealing to try to capture such filamentary structures and to
approximate the data by metric graphs that will summarize the metric and allow
convenient visualization. 

\paragraph{Contribution} 
In this paper we address the metric reconstruction problem for filamentary
structures. The input of our method and algorithm is a metric space $(X, d_X)$
that is assumed to be close with respect to the so-called Gromov-Hausdorff
distance $d_{GH}$ to a much simpler, but unknown, metric graph $(G', d_{G'})$. Our
algorithm outputs a metric graph $(G, d_G)$ that is proven to be close to $(X,
d_{X})$. Our approach relies on the notion of Reeb graph (and some variants of
it introduced in Section \ref{sec:reeb-graph}) and one of our main theoretical
result can be stated as follows.  
\vskip5pt\noindent
{\bf Theorem \ref{thm:bound-dGH-final}}. 
{\em Let $(X, d_X)$ be a compact connected geodesic space, let $r \in X$ be a fixed base point such that the metric Reeb graph $(G,d_G)$ of the function $d = d_X(r,.): X \rightarrow \mathbb{R}$ is a finite graph.
If for a given $\e >0$ there exists a finite metric graph $(G',d_{G'})$ such that $d_{GH}(X,G') < \e$ then we have 
\[
d_{GH}(X, G) < 2 (\beta_1(G) + 1)  (17 + 8 N_{E,G'}(8 \e)) \e
\] 
where $N_{E,G'}(8\e)$ is the number of edges of $G'$ of length at most $8\e$ and $\beta_1(G)$ is the first Betti number of $G$, i.e. the number of edges to remove from $G$ to get a spanning tree. 
In particular if $X$ is at distance less than $\e$ from a metric graph with shortest edge larger than $8\e$ then $d_{GH}(X,G) < 34  (\beta_1(G) + 1) \e$.}
\vskip5pt\noindent
Turning this result into a practical algorithm requires to address two issues: 
\begin{itemize}
\item[-] First, raw data usually do not come as geodesic spaces. They are given as
  discrete sets of point (and thus not connected metric spaces) sampled from
  the underlying space $(X, d_X)$. Moreover in many cases only distances
  between nearby points are known. A geodesic space (see Section
  \ref{sec:metric-GH} for a definition of geodesic space) can then be obtained
  from these raw data as a neighborhood graph where nearby points are connected
  by edges whose length is equal to their pairwise distance. The shortest path
  distance in this graph is then used as the metric. In our experiments we use this new metric as the input of our algorithm. 
The question of the approximation of the metric on $X$ by the metric induced on
  the neighborhood graphs is out of the scope of this paper.
\item[-] Second, approximating the Reeb graph $(G,d_G)$ from a neighborhood graph is
  usually not obvious. If we compute the Reeb graph of the distance function to
  a given point defined on the neighborhood graph we obtain the neighborhood
  graph itself and do not achieve our goal of representing the input data by a
  simple graph. It is then appealing to build a two dimensional complex having
  the neighborhood graph as $1$-dimensional skeleton and use the algorithm of
  \cite{hww-rtac-10, Parsa:2012} to compute the Reeb graph of the distance to the root
  point. Unfortunately adding triangles to the neighborhood graph may widely
  change the metric between the data points on the resulting complex and
  significantly increase the complexity of the algorithm. We
  overcome this issue by introducing a variant of the Reeb graph, the
  $\alpha$-Reeb graph, inspired from \cite{smc-tmah-07} and related to the recently introduced notion of graph induced complex \cite{dfw-gicpd-13}, that is easier to
  compute than the Reeb graph but also comes with approximation guarantees
  (see Theorem \ref{thm:bound-dGH-final-alpha}). 
As a consequence our algorithm
  relies on the Mapper algorithm of \cite{smc-tmah-07} and runs in almost
  linear time (see Section~\ref{sec:algo-pl}).
\end{itemize}

\paragraph{Related work.} 
Approximation of data by $1$-dimensional geometric structures has been
considered by different communities.  In statistics, several approaches have
been proposed to address the problem of detection and extraction of filamentary
structures  in point cloud data.  For example Arial-Castro et al
\cite{ariascastro06} use  multiscale anisotropic strips to detect linear
structure while \cite{genovese09,gppvw-gnfe-12} and more recently
\cite{gppvw-nre-12} base their approach upon density gradient descents or
medial axis techniques. These methods apply to data corrupted by outliers
embedded  in Euclidean spaces and focus on the inference of individual
filaments without focus on the global geometric structure of the filaments
network. 

In computational geometry, the curve reconstruction problem from
points sampled on a curve in an euclidean space has been extensively studied
and several efficient algorithms have been proposed \cite{ABE98,DMR00,DW01}.
Unfortunately, these methods restricts to the case of simple embedded curves
(without singularities or self-intersections) and hardly extend to the case of
topological graphs.  In a more intrinsic setting where data come as finite
abstract metric spaces,
\cite{MRIDUL-metric-graph} propose an algorithm that outputs, under some
specific sampling conditions, a topologically correct (up to a homeomorphism)
reconstruction of the approximated graph. However  this algorithm requires some
tedious parameters tuning and relies on quite restrictive sampling assumptions.
When these conditions are not satisfied, the algorithm may fail and not even
outputs a graph.  Compared to the algorithm of \cite{MRIDUL-metric-graph}, our
algorithm not only comes with metric guarantees but also whatever the input data
is, it always outputs a metric graph and does not require the user to choose any
parameters. Our approach is also related to the so-called Mapper algorithm
\cite{smc-tmah-07} that provides a way to visualize data sets endowed with a
real valued function as a graph. Indeed the implementation of our algorithm
relies on the Mapper algorithm where the considered function is the distance to
the chosen root point. However, unlike the general mapper algorithm, our
methods provides an upper bound on the Gromov-Hausdorff distance between the
reconstructed graph and the underlying space from which the data points have
been sampled.  

In theoretical computer science, there is much of work on approximating
metric spaces using trees \cite{Badoiu:2007:AAE, AbrahamBKMRT07, ChepoiDEHV08} or 
distribution of trees \cite{Dhamdhere:2006,Fakcharoenphol:2003} where
the trees are often constructed as spanning trees possibly with Steiner points.
Our approach is different as our reconstructed graph or tree is a quotient space
of the original metric space where the metric only gets contracted 
(see Lemma~\ref{lemma:pi-Lipschitz}). 
Finally we remark that the recovery of filament structure is also studied in 
various applied settings, including road networks \cite{CGHS10,tupin98}, 
galaxies distributions \cite{2010MNRAS.tmp..692C}. 

\vskip10pt 
The paper is organized as follows.  The basic notions and definitions used throughout the paper are recalled in Section \ref{sec:metric-GH}. The Reeb and $\alpha$-Reeb graphs endowed with a natural metric are introduced in Section \ref{sec:reeb-graph} and the approximation results are stated and proven in Sections \ref{sec:reeb-bound} and \ref{sec:M-bound}. Our algorithm is described in Section \ref{sec:algo-pl} and  experimental results are presented and discussed in Section \ref{sec:exp}.


\section{Preliminaries: metric spaces, metric graphs and Gromov-Hausdorff distance}  \label{sec:metric-GH}
%
Recall that a metric space is a pair $(X,d_X)$ where $X$ is a set and $d_X : X \times X \to \R$ is a non negative map such that for any $x,y,z \in X$, $d_X(x,y) = 0$ if and only if $x=y$, $d_X(x,y) = d_X(y,x)$ and $d_X(x,z) \leq d_X(x,y) + d_X(y,z)$.
Two compact spaces $(X,d_X)$ and $(Y,d_Y)$ are isometric if there exits a bijection $\phi : X \to Y$ that preserves the distances, namely: for any $x,x' \in X, d_Y(\phi(x),\phi(x')) = d_X(x,x')$. The set of isometry classes of compact metric spaces can be endowed with a metric, the so-called Gromov-Hausdorff distance that can be defined using the following notion of correspondence (\cite{bbi-cmg-01} Def. 7.3.17).

\begin{definition}
Let $(X,d_X)$ and $(Y,d_Y)$ be two compact metric spaces. Given $\e >0$, an {\em $\e$-correspondence} between $(X,d_X)$ and $(Y,d_Y)$ is a subset $C \subset X \times Y$ such that:\\
i) for any $x \in X$ there exists $y \in Y$ such that $(x,y) \in C$;\\
ii) for any $y \in Y$ there exists $x \in X$ such that $(x,y) \in C$;\\
iii) for any $(x,y), (x',y') \in C$, $| d_X(x,x') - d_Y(y,y')| \leq \e$.
\end{definition}

\begin{definition}
The {\em Gromov-Hausdorff distance} between two compact metric spaces $(X,d_X)$ and $(Y,d_Y)$ is defined by
$$d_{GH}(X,Y) = \frac{1}{2} \inf \{\e \geq 0 : \mbox{\rm there exists an $\e$-correspondence between $X$ and $Y$} \}$$
\end{definition}

A metric space $(X,d_X)$ is a {\em path metric space} if the
distance between any pair of points is equal to the infimum of the lengths of
the continuous curves joining them \footnote{see \cite{g-msrnr-07} Chap.1 for
the definition of the length of a continuous curve in a general metric space}.
Equivalently $(X,d_X)$ is a path metric space if and only if for any $x,y \in
X$ and any $\varepsilon >0$ there exists $z \in X$ such that $\max(d_X(x,z),
d_X(y,z)) \leq \frac{1}{2} d_X(x,y) + \varepsilon$ \cite{g-msrnr-07}.  In the
sequel of the paper we consider compact path metric spaces. It follows from the
Hopf-Rinow theorem (see \cite{g-msrnr-07} p.9) that such spaces are {\em
geodesic}, i.e. for any pair of point $x, x' \in X$ there exists a minimizing
geodesic joining them.\footnote{recall that a minimizing geodesic in $X$ is any
curve $\gamma : I \to X$, where $I$ is a real interval, such that
$d_X(\gamma(t),\gamma(t')) = |t - t'|$ for any $t,t' \in I$.}
A continuous path $\delta : I \to X$ where $I$ is a real interval or the unit circle is said
to be {\em simple} if it is not self intersecting, i.e. if $\delta$ is an
injective map.

Recall that a {\em (finite) topological graph} $G = (V,E)$ is the geometric
realisation of a (finite) $1$-dimensional simplicial complex with vertex set
$V$ and edge set $E$. If moreover each $1$-simplex $e \in E$  is a metric edge,
i.e. $e = [a,b] \subset \mathbb{R}$, then the graph $G$ inherits from a metric
$d_G$ which is the unique one whose restriction to any $e = [a,b] \in E$
coincides with the standard metric on the real segment $[a,b]$.  Then $(G,
d_G)$ is a {\em metric graph} (see \cite{bbi-cmg-01}, Section 3.2.2 for a more
formal definition). Intuitively, a metric graph can be seen as a topological
graph with a length assigned to each of its edges. 

The {\em first Betti 
number $\beta_1(G)$} of a finite topological graph $G$ is the rank of the first
homology group of $G$ (with coefficient in a field, e.g. $\mathbb{Z}/2$), or
equivalently,  the number of edges to remove from $G$ to get a spanning tree. 


\section{Approximation of path metric spaces with Reeb-like graphs}  \label{sec:approx-reeb}

Let $(X, d_X)$ be a compact geodesic space and let $r \in X$ be a fixed base point. Let $d: X \rightarrow \mathbb{R}$
be the distance function to $r$, i.e., $d(x) = d_X(r, x)$. 

\subsection{The Reeb and $\alpha$-Reeb graphs of $d$}  \label{sec:reeb-graph}
\paragraph{The Reeb graph.}
The relation $x \sim y$ if and only if $d(x) = d(y)$ and $x, y$ are in the same path connected component of $d^{-1}(d(x))$ is an equivalence relation. The quotient space $G = X/\sim$ is called the {\em Reeb graph} of $d$ and we denote by $\pi : X \to G$ the quotient map. Notice that $\pi$ is continuous and as $X$ is path connected, $G$ is path connected. The function $d$ induces a function $d_* : G \to \R_+$ that satisfies $d = d_* \circ \pi$. 
The relation defined by: for any $g,g' \in G$, $g \leq_G g'$ if and only if $d_*(g) \leq d_*(g')$ and there exist a continuous path $\gamma$ in $G$ connecting $g$ to $g'$ such that $d \circ \gamma$ is non decreasing, makes $G$ a partially ordered set. 

\paragraph{The $\alpha$-Reeb graphs.}
Computing or approximating the Reeb graph of $(X,d)$ from a finite set of point sampled on $X$ is usually a difficult task. To overcome this issue we also consider a variant of the Reeb graph that shares very similar properties than the Reeb graph. Let $\alpha > 0$ and let $\mathcal{I} = \{ I_i \}_i \in I$ be a covering of the range of $d$ by open intervals of length at most $\alpha$. 
The transitive closure of the relation $x \sim_\alpha y$ if and only if $d(x) = d(y)$ and $x, y$ are in the same path connected component of $d^{-1}(I_i)$ for some interval $I_i \in \mathcal{I}$ is an equivalence relation that is also denoted by $\sim_\alpha$. The quotient space $G_\alpha = X/\sim_\alpha$ is called the {\em $\alpha$-Reeb graph}\footnote{strictly speaking we should call it the $\alpha$-Reeb graph associated to the covering $\mathcal{I}$ but we assume in the sequel that some covering $\mathcal{I}$ has been chosen and we omit it in notations} of $d$ and we denote by $\pi : X \to G_\alpha$ the quotient map. 
Notice that $\pi$ is continuous and as $X$ is path connected, $G_\alpha$ is path connected. The function $d$ induces a function $d_* : G_\alpha \to \R_+$ that satisfies $d = d_* \circ \pi$. The relation defined by: for any $g,g' \in G_\alpha$, $g \leq_{G_\alpha} g'$ if and only if $d_*(g) \leq d_*(g')$ and there exist a continuous path $\gamma$ in $G_\alpha$ connecting $g$ to $g'$ such that $d \circ \gamma$ is non decreasing, makes $G_\alpha$ a partially ordered set. 

The $\alpha$-Reeb graph is closely related to the graph constructed by the Mapper algorithm introduced in \cite{smc-tmah-07} making its computation much more easier than the Reeb graph (see Section \ref{sec:algo-pl}). 
\vskip10pt\noindent
Notice that without making assumptions on $X$ and $d$, in general $G$ and $G_\alpha$ are not finite graphs. However when the number of path connected components of the level sets of $d$ is finite and changes only a finite number of times then the Reeb graph turns out to be a finite directed acyclic graph. Similarly, when the covering of $X$ by the connected components of $d^{-1}(I_i), i \in \mathcal{I}$ is finite, the $\alpha$-Reeb graph also turns out to be a finite directed acyclic graph.
This happens in most applications and for example when $(X,d_X)$ is a finite simplicial complex or a compact semialgebraic (or more generally a compact subanalytic space) with $d$ being semi-algebraic (or subanalytic). 
\vskip10pt\noindent
All the results and proofs presented in Section \ref{sec:approx-reeb} are exactly the same for the Reeb and the $\alpha$-Reeb graphs. In the following paragraph and in Section \ref{sec:reeb-bound}, $G$ denotes indifferently the Reeb graph or an $\alpha$-Reeb graph for some $\alpha >0$. We also always assume that $X$ and $d$ (and $\alpha$ and $\mathcal{I}$) are such that $G$ is a finite graph.  


\paragraph{A metric on Reeb and $\alpha$-Reeb graphs.} Let define the set of vertices $V$ of $G$ as the union of the set of points of degree not equal to $2$ with the set of local maxima of $d_*$ over $G$, and the base point $\pi(r)$. The set of edges $E$ of $G$ is then the set of the connected components of the complement of $V$. Notice that $\pi(r)$ is the only local (and global) minimum of $d_*$: since $X$ is path connected, for any $x \in X$ there exists a geodesic $\gamma$ joining $r$ to $x$ along which $d$ is increasing; $d_*$ is thus also increasing along the continuous curve $\pi(\gamma)$, so $\pi(x)$ cannot be a local minimum of $d_*$. As a consequence $d_*$ is monotonous along the edges of $G$. We can thus assign an orientation to each edge: if $e = [p,q] \in G$ is such that $d_*(p) < d_*(q)$ then the positive orientation of $e$ is the one pointing from $p$ to $q$. 
Finally, we assign a metric to $G$. Each edge $e\in E$ is homeomorphic to an interval to which we assign a length equal to the absolute difference of the function $d_*$ at two endpoints. The distance between two points $p,p'$ of $e$ is then $|d_*(p) - d_*(p')|$. 
This makes $G$ a metric graph $(G, d_G)$ isometric to the quotient space of
the union of the intervals isometric to the edges by identifying the endpoints if they correspond to the 
same vertex in $G$. See Figure~\ref{fig:reebgraph} for an example. Note that $d_*$ is continuous in $(G, d_G)$ and for any $p \in G$, $d_*(p) = d_G(\pi (r), p)$. Indeed this is a consequence of the following lemma.

\begin{lemma} \label{lemma:geodesic-G}
If $\delta$ is a path joining two points $p,p' \in G$ such that $d_* \circ \delta$ is strictly increasing then $\delta$ is a shortest path between $p$ and $p'$ and $d_G(p,p') = d_*(p') - d_*(p)$. 
\end{lemma}

\begin{proof}
As $d_* \circ \delta$ is strictly increasing, when $\delta$ enters an edge $e$ by one of its end points, either it exits at the other end point or it stops at $p'$ if $p' \in e$. Moreover $\delta$ cannot go through a given edge more than one time. As a consequence $\delta$ can be decomposed in a finite sequence of pieces $e_0 = [p, p_1], e_1 = [p_1,p_2], \cdots, e_{n-1} = [p_{n-1}, p_{n}], e_n = [p_n, p']$ where $e_0$ and $e_n$ are the segments joining $p$ and $p'$ to one of the endpoint of the edges that contain them and $e_1, \cdots, e_{n-1}$ are edges. So, the length of $\delta$ is equal to $(d_*(p_1)-d_*(p)) + (d_*(p_2) - d_*(p_1)) + \cdots + (d_*(p') - d_*(p_n)) = d_*(p') - d_*(p)$ and $d_G(p,p') \leq d_*(p') - d_*(p)$. 

Similarly any simple path joining $p$ to $p'$ can be decomposed in a finite sequence of pieces $e'_0 = [p, p'_1], e'_1 = [p'_1,p'_2], \cdots, e'_{k-1} = [p'_{k-1}, p'_{k}], e'_k = [p'_k, p']$ where $e'_0$ and $e'_k$ are the segments joining $p$ and $p'$ to one of the endpoint of the edges that contain them, and $e'_1, \cdots, e'_{k-1}$ are edges. Now, as we do not know that $d_*$ is increasing along this path, its length is thus equal to $|d_*(p'_1)-d_*(p)| + |d_*(p'_2) - d_*(p'_1)| + \cdots + |d_*(p') - d_*(p'_n)| \geq d_*(p') - d_*(p)$. So, $d_G(p,p') \geq d_*(p') - d_*(p)$.
\end{proof}

\begin{figure}[!h]
\centering
\includegraphics[height=0.3\textwidth]{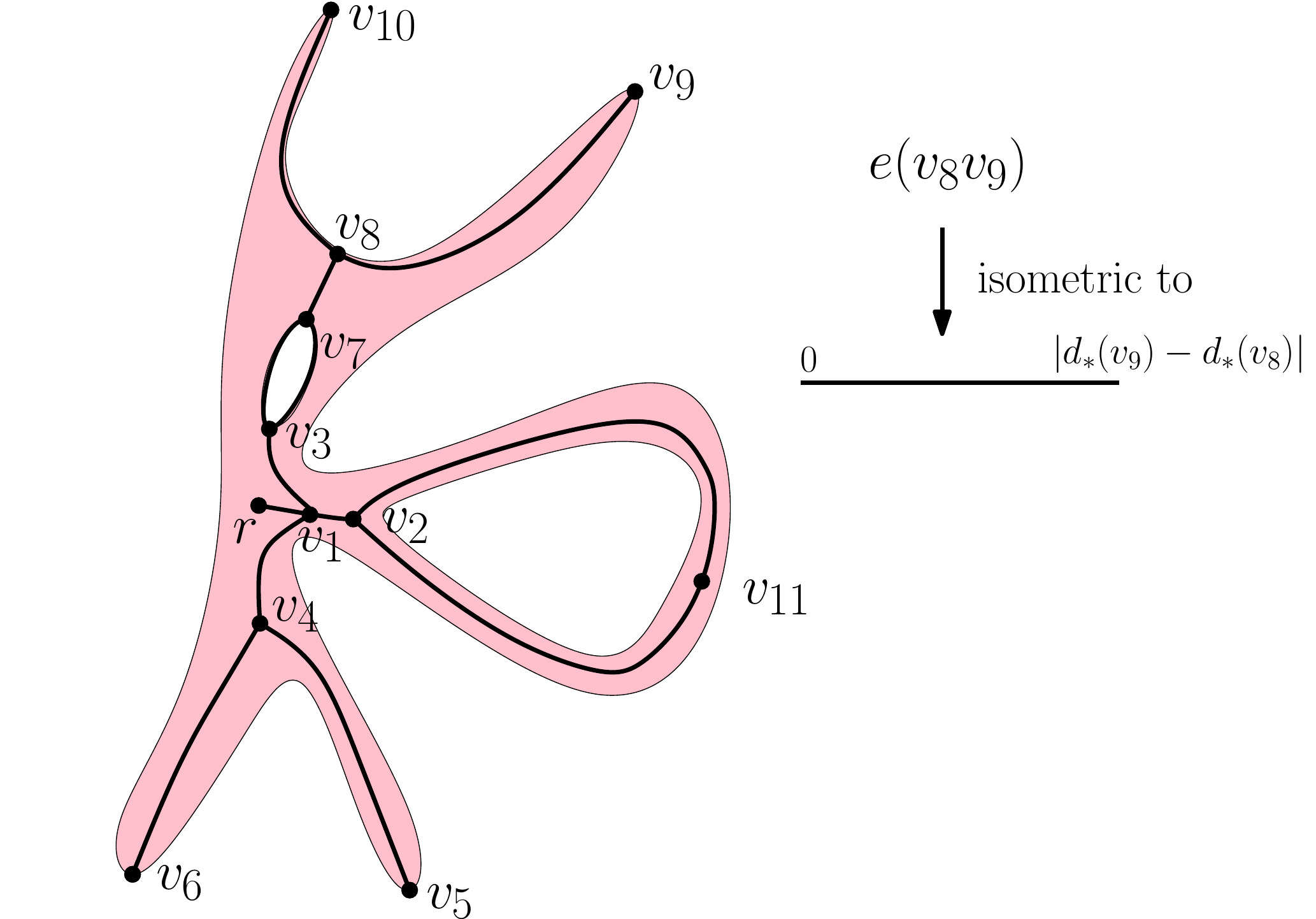} 
\caption{
Construction of a Reeb graph.
\label{fig:reebgraph}}
\end{figure}

\subsection{Bounding the Gromov-Hausdorff distance between $X$ and $G$}  \label{sec:reeb-bound}
The goal of this section is to provide an upper bound of the Gromov-Hausdorff distance between $X$ and $G$ that only depends on the first Betti number $\beta_1(G)$ of $G$ and the maximal diameter $M$ of the level sets of $\pi$. An upper bound of $M$ is given in the next section. 

\begin{theorem} \label{thm:bound-dGH-betti1-M}
$d_{GH}(X, G) < (\beta_1(G) + 1) M$ where $d_{GH}(X, G)$ is the Gromov-Hausdorff distance between $X$ and $G$, $\beta_1(G)$ is the first Betti number of $G$ and $M = \sup_{p\in G}\{\text{diameter}(\pi^{-1}(p))\}$ is the supremum of the diameters of the level sets of $\pi$.
\theolab{thm:bound}
\end{theorem}

The proof of Theorem \ref{thm:bound-dGH-betti1-M} will be easily deduced from a sequence of technical lemmas that allow to compare the distance between pair of points $x,y \in X$ and the distance between their images $\pi(x), \pi(y) \in G$. 



A vertex $v \in V$ is called a {\em merging vertex} if it is the end point of at least two edges $e_1$ and $e_2$ that are pointing to it according to the orientation defined in Section \ref{sec:reeb-graph}. Geometrically this means that there are at least two distinct connected components of $\pi^{-1}(d_*^{-1}(d_*(v) - \e))$ that accumulate to $\pi^{-1}(v)$ as $\e>0$ goes to $0$. The set of merging vertices is denoted by $V_m$.

The following lemma provides an upper bound on the number of vertices in $V_m$.
\begin{lemma} \label{lem:mv}
The number of elements in $V_m$ is equal to $\beta_1(G)$ where $\beta_1(G)$ is the first homology group of $G$.
\end{lemma}
\begin{proof}
The result follows from classical homology persistence theory \cite{e-cti-10}. First remark that, as $\pi(r)$ is the only local minimum of $d_*$, the sublevel sets of the function $d_* : G \to \R_+$ are all path connected. Indeed if $\pi(x), \pi(y) \in G$ are in the same sublevel set $d_*^{-1}([0,\alpha])$, $\alpha >0$, then the images by $\pi$ of the shortest paths in $X$ connecting $x$ to $r$ and $y$ to $r$ are contained in $d_*^{-1}([0,\alpha])$ and their union is a continuous path joining $\pi(x)$ to $\pi(y)$.
As a consequence, the $0$-dimensional persistence of $d_*$ is trivial. So, for the persistence algorithm applied to $d_*$, the vertices of $V_m$ are the positive simplices (see \cite{elz-tps-02} for the notion of positive and negative simplices for persistence). It follows that each vertex of $V_m$ creates a cycle that never dies as $G$ is one dimensional and does not contain any $2$-dimensional simplex. Thus $|V_m| = \beta_1(G)$. 
\end{proof}

\begin{lemma} \label{lemma:path-out-Vm}
Let $p, p' \in G$ and let $\delta : [d_*(p), d_*(p')] \to G$ be a strictly increasing path going from $p$ to $p'$ that does not contain any point of $V_m$ in its interior. Then for any $x' \in \pi^{-1}(p') \cap cl(\pi^{-1}(\delta(d_*(p), d_*(p')))$ where $cl(.)$ denotes the closure, there exists a shortest path $\gamma$ connecting a point $x$ of $\pi^{-1}(p)$ to $x'$ such that $\pi(\gamma)= \delta$ and $d_X(x,x') = d(x') - d(x) = d_*(p') - d_* (p) = d_G(p,p')$. 
\end{lemma}

Notice that from Lemma \ref{lemma:geodesic-G}, $\delta$ is a shortest path and the parametrization by the interval $[d_*(p), d_*(p')]$ can be chosen to be an isometric embedding.

\begin{proof}
First assume that $p'$ is not a merging point. 
Let $\gamma_0 : [0,d(x')] \to X$ be any shortest path between $r$ and $x'$ and let $\gamma$ be the restriction of $\gamma_0$ to $[d_*(p),d(x')] = [d_*(p),d_*(p')]$. 
If the infimum $t_0$ of the set $I = \{ t \in [d_*(p),d_*(p')] : \pi(\gamma(t')) \in \delta,\ \forall t' \geq t \}$ is larger than $d_*(p)$, then $\pi(\gamma(t_0))$ then there exists an increasing sequence $(t_n)$ that converges to $t_0$ such that $\gamma(t_n) \not \in \delta$. As a consequence $\delta(t_0)$ is a merging point; a contradiction. So $t_0 = d_*(p)$ and $\gamma(d_*(p))$ intersects $\pi^{-1}(p)$ at a point $x$. 

Now if $p'$ is a merging point, as $x'$ is chosen in the closure of $\pi^{-1}(\delta(d_*(p), d_*(p'))$, for any sufficiently large $n \in \mathbb{N}$ one can consider a sequence of points $x'_n \in \pi^{-1}(\delta(d_*(p') - 1/n))$ that converges to $x'$ and apply the first case to get a sequence of shortest path $\gamma_n$ from a point $x_n \in \pi^{-1}(p)$ and $x'_n$. Then applying Arzel\`a-Ascoli's theorem (see \cite{die-69} 7.5) we can extract from $\gamma_n$ a sequence of points converging to a shortest path $\gamma$ between a point $x \in \pi^{-1}(p)$ and $x'$.

To conclude the proof, notice that from Lemma \ref{lemma:geodesic-G} we have $d_G(p,p') = d_*(p') - d_* (p) = d(x') - d(x)$. Since $\gamma$ is the restriction of a shortest path from $r$ to $x$ we also have $d_X(x,x') = d(x') - d(x)$.
\end{proof}

Lemma \ref{lemma:bound-dX} and Lemma \ref{lemma:pi-Lipschitz} allow to compare the metrics $d_X$ and $d_G$. 

\begin{lemma} \label{lemma:bound-dX}
For any $x, y \in X$ we have 
\[
d_X(x,y) \leq d_G(\pi(x),\pi(y)) + 2 (\beta_1(G)+1) M
\]
where $\beta_1(G)$ is the first Betti number of $G$ and $M = \sup_{p\in G}\{\text{diameter}(\pi^{-1}(p))\}$.
\end{lemma}

\begin{proof}
Let $\delta$ be a shortest path between $\pi(x)$ and $\pi(y)$. 
Remark that except at the points $\pi(x)$ and $\pi(y)$ the local maxima of the restriction of $d_*$ to $\delta$ are in $V_m$. Indeed as
$\delta$ is a shortest path it has to be simple, so if $p \in \delta$ is a local maximum then $p$ has to be a vertex and $\delta$ has to pass through two edges having $p$ as end point and pointing to $p$ according to the orientation defined in Section \ref{sec:reeb-graph}. So $p$ is a merging point. \\
Since $\delta$ is simple and $V_m$ is finite, $\delta$ can be decomposed in at most $|V_m|+1$ connected paths along the interior of which the restriction of $d_*$ does not have any local maxima. So along each of these connected paths the restriction of $d_*$ can have at most one local minimum. As a consequence, $\delta$ can be decomposed in a finite number of continuous paths $\delta_1, \delta_2, \cdots, \delta_k$ with $k \leq 2(|V_m|+1)$, such that the restriction of $d_*$ to each of these path is strictly monotonous. 
For any $i \in \{1, \cdots, k \}$ let $p_i$ and $p_{i+1}$ the end points of $\delta_i$ with $p_1 = \pi(x)$ and $p_{k+1} = \pi(y)$.
We can apply Lemma \ref{lemma:path-out-Vm} to each $\delta_i$ to get a shortest path $\gamma_i$ in $X$  between a point $x_i \in \pi^{-1}(p_i)$ and a point in $y_{i+1} \in \pi^{-1}(p_{i+1})$ such that $\pi(\gamma_i) = \delta_i$ and $d_X(x_i,y_{i+1}) = d_G(p_i,p_{i+1})$. The sum of the lengths of the paths $\gamma_i$ is equal to the sum of the lengths of the path $\delta_i$ which is itself equal to $d_G(\pi(x),\pi(y))$. 
Now for any $i \in \{1, \cdots, k \}$, since $\pi(x_i) = \pi(y_i)$ we have $d_X(x_i, y_i) \leq M$ and $x_i$ and $y_i$ can be connected by a path of length at most $M$ ($x_1$ is connected to $x$ and $y_{k+1}$ is connected to $y$. Gluing these paths to the paths $\gamma_i$ gives a continuous path from $x$ to $y$ whose length is at most $d_G(\pi(x),\pi(y)) + k M \leq d_G(\pi(x),\pi(y)) + 2(|V_m|+1)M$.
Since from Lemma \ref{lem:mv}, $|V_m| \leq \beta_1(G)$, we finally get that $d_X(x,y) \leq  d_G(\pi(x),\pi(y)) + 2(\beta(G)+1)M$.
\end{proof}

\begin{lemma} \label{lemma:pi-Lipschitz}
The map $\pi : X \to G$ is $1$-Lipschitz: for any $x, y \in X$ we have
\[
d_G(\pi(x),\pi(y)) \leq d_X(x,y).
\]
\end{lemma}

\begin{proof}
Let $x,y \in X$ and let $\gamma : I \to X$ be a shortest path from $x$ to $y$ in $X$ where $I \subset \R$ is a closed interval. 
The path $\pi(\gamma)$ connects $\pi(x)$ and  $\pi(y)$ in $G$. 

We first claim that there exists a continuous path $\Gamma$ contained in $\pi(\gamma)$ connecting  $\pi(x)$ and  $\pi(y)$ that intersects each vertex of $G$ at most one time. The path $\Gamma$ can be defined by iteration in the following way. Let $v_1, \cdots v_n \in V$ be the vertices of $G$ that are contained in $\pi(\gamma) \setminus \{ \pi(x), \pi(y) \}$ and let $\Gamma_0 = \pi(\gamma) : J_0=I \to G$. For $i=1, \cdots n$ let  $t_i^- = \inf \{ t : \Gamma_{i-1}(t) = v_i \}$ and  $t_i^+ = \sup \{ t : \Gamma_{i-1}(t) = v_i \}$ and define $\Gamma_i$ as the restriction of $\Gamma_{i-1}$ to $J_i = J_{i-1} \setminus (t_i^-, t_i^+)$. The path $\Gamma_i$ is a connected continuous path (although $J_i$ is a disjoint union of intervals) that intersects the vertices $v_1, v_2, \cdots, v_i$ at most one time. We then define $\Gamma = \Gamma_n : J= J_n \to G$ where $J \subset I$ is a finite union of closed intervals. Notice that $\Gamma$ is the image by $\pi$ of the restriction of $\gamma$ to $J$ and that $\Gamma(t) \in \{ v_1, \cdots v_n \}$ only if $t$ is one of the endpoints of the closed intervals defining $J$. 

Now, for each connected component $[t,t']$ of $J$, $\gamma((t,t'))$ is contained in $\pi^{-1}(e)$ where $e$ is the edge of $G$ containing $\Gamma ([t,t'])$. As a consequence, $d_G(\pi(\gamma)(t),\pi(\gamma)(t')) = |d_*(\pi(\gamma)(t) - d_*(\pi(\gamma)(t'))| = |d(\gamma(t)) - d(\gamma(t'))|$. 
Recalling that $d(\gamma(t)) = d_X(r,\gamma(t))$ and $d(\gamma(t')) = d_X(r,\gamma(t'))$ and using the triangle inequality we get that $|d(\gamma(t)) - d(\gamma(t'))| \leq d_X(\gamma(t),\gamma(t'))$. 
To conclude the proof, since $\gamma$ is a geodesic path we just need to sum up the previous inequality over all connected components of $J$:
\[
d_X(x,y) \geq \sum_{[t,t'] \in cc(J)} d_X(\gamma(t),\gamma(t')) \geq \sum_{[t,t'] \in cc(J)} d_G(\pi(\gamma)(t),\pi(\gamma)(t')) \geq d_G(\pi(x), \pi(y))
\]
where $cc(J)$ is the set of connected components of $J$.
\end{proof}

The proof of Theorem \ref{thm:bound-dGH-betti1-M} now easily follows from Lemmas \ref{lemma:bound-dX} and \ref{lemma:pi-Lipschitz}.

\begin{proof} {\em (of Theorem \ref{thm:bound-dGH-betti1-M})}
Consider the set $C = \{ (x,\pi(x)): x \in X \} \subset X \times G$. As $\pi$ is surjective this is a correspondence between $X$ and $G$. It follows from Lemmas \ref{lemma:bound-dX} and \ref{lemma:pi-Lipschitz} that for any $(x,\pi(x)), (y,\pi(y)) \in C$, 
\[
|d_X(x,y) - d_G(\pi(x),\pi(y))| \leq 2 (\beta_1(G)+1) M
\]
where $\beta_1(G)$ is the first Betti number of $G$ and $M = \sup_{p\in G}\{\text{diameter}(\pi^{-1}(p))\}$. So $C$ is a $(2 (\beta_1(G)+1) M$-correspondence and $d_{GH}(X,G) \leq (\beta_1(G)+1) M$.
\end{proof}

\subsection{Bounding $M$}  \label{sec:M-bound}

To upperbound the diameter of the level sets of $\pi$ we first prove the two following general lemmas. 

\begin{lemma} \label{lemma:edge-distance-root}
Let $(G, d_G)$ be a connected finite metric graph and let $r \in G$. We denote by $d_r = d_G(r,.) : G \to [0,+\infty)$ the distance to $r$. 
For any edge $E \subset G$, the restriction of $d_r$ to $e$ is either strictly monotonous or has only one local maximum.
Moreover the length $l = l(E)$ of $E$ is upper bounded by two times the difference between the maximum and the minimum of $d_r$ restricted to $E$.
\end{lemma}

\begin{proof}
Let $l$ be the length of $E$ and let $t \mapsto e(t)$, $t \in [0,l]$, be an arc length parametrization of $E$.
Since $E$ is an edge of $G$, for $t \in [0,l]$ any shortest geodesic $\gamma_t$ joining $r$ to $e(t)$ must contain either $x_1 = e(0)$ or $x_2 = e(l)$. If it contains $x_1$ then for any $t'<t$ the restriction of $\gamma_t$ between $r$ and $e(t')$ is a shortest geodesic containing $x_1$ and if it contains $x_2$ then for any $t'>t$ the restriction of $\gamma_t$ between $r$ and $e(t')$ is a shortest geodesic containing $x_2$. Moreover in both cases, the function $d_r$ is strictly monotonous along $\gamma$.
As a consequence, the set $I_1 = \{ t \in [0,l]:$~a shortest geodesic joining $r$ to $e(t)$ contains $x_1 \}$ is a closed interval containing $0$. Similarly the set $I_2 = \{ t \in [0,l]:$~a shortest geodesic joining $r$ to $e(t)$ contains $x_2 \}$ is a closed interval containing $l$ and $[0,l] = I_1 \cup I_2$. Moreover $d_r$ is strictly monotonous on $e(I_1)$ and on $e(I_2)$.
As a consequence $I_1 \cap I_2$ is reduced to a single point $t_0$ that has to be the unique local maximum of $d_r$ restricted to $E$. 

The second part of the lemma follows easily from the previous proof: the minimum of $d_r$ restricted to $E$ is attained either at $x_1$ or $x_2$ and $d_r(e(t_0)) = d_r(x_1)+t_0 = d_r(x_2)+l-t_0$ is the maximum of $d_r$ restricted to $E$. We thus obtain that $2 t_0 = l + (d_r(x_2) -d_r(x_1))$. As a consequence if $d_r(x_1) \leq d_r(x_2)$ then $l/2 \leq t_0 = d_r(e(t_0)) - d_r(x_1)$; similarly if  $d_r(x_1) \geq d_r(x_2)$ then $l/2 \leq l - t_0 = d_r(e(t_0)) - d_r(x_2)$.
\end{proof}

\begin{lemma} \label{proposition:band-diameter}
Let $(G, d_G)$ be a connected finite metric graph and let $r \in G$. For $\alpha > 0$ we denote by $N_E(\alpha)$ the number of edges of $G$ of length at most $\alpha$.  
For any $d >0$ and any connected component $B$ of the set $B_{d,\alpha} = \{ x \in G: d-\alpha \leq d_G(r,x) \leq d+\alpha \}$ we have
\[
diam(B) \leq 4 (2+N_E(4 \alpha)) \alpha
\]
\end{lemma}

The example of figure \ref{fig:diambound} shows that the bound of Lemma \ref{proposition:band-diameter} is tight. 

\begin{figure}[!h]
\centering
\includegraphics[height=0.3\textwidth]{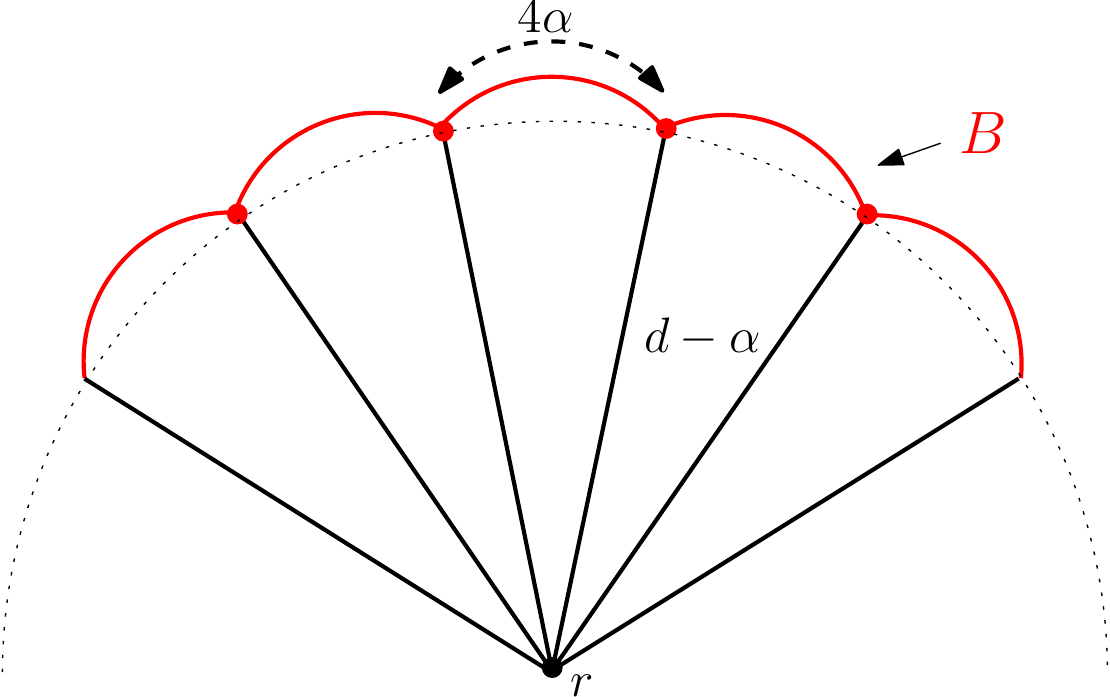} 
\caption{Tightness of the bound in Lemma \ref{proposition:band-diameter}: there are 3 edges of length at most $4\alpha$ and the diameter of $B$ is equal to $20 \alpha$. 
\label{fig:diambound}}
\end{figure}

\begin{proof}
Let $x,y \in B$ and let $t \mapsto \gamma(t) \in B$ be a continuous path joining $x$ to $y$ in $B$.
Let $E$ be an edge of $G$ that does not contain $x$ or $y$ and with end points $x_1,x_2$ such that $\gamma$ intersects the interior of $E$. Then $\gamma^{-1}(E)$ is a disjoint union of closed intervals of the form $I = [t,t']$ where $\gamma(t)$ and $\gamma(t')$ belong to the set $\{ x_1, x_2 \}$. If $\gamma(t) = \gamma(t')$ we can remove the part of $\gamma$ between $t$ and $t'$ and still get a continuous path between $x$ and $y$. So without loss of generality we can assume that if $\gamma$ intersects the interior of $E$, then $E$ is contained in $\gamma$. Using the same argument as previously we can also assume that if $\gamma$ goes across $E$, it only does it one time, i.e. $\gamma^{-1}(E)$ is reduced to only one interval.
As a consequence, $\gamma$ can be decomposed in a sequence $[x,v_0], E_1, E_2, \cdot, E_k, [v_k,y]$ where $[x,v_0]$ and $[v_k,y]$ are pieces of edges containing $x$ and $y$ respectively and $E_1 = [v_0,v_1], E_2=[v_1,v_2] \cdot, E_k = [v_{k-1},v_k]$ are pairwise distinct edges of $G$ contained in $B$. 
It follows from Lemma \ref{lemma:edge-distance-root} that the lengths of the edges $E_1, \cdots E_k$ and of $[x,v_0]$ and $[v_k,y]$ are upper bounded by $4\alpha$. As a consequence the length of $\gamma$ is upper bounded by $4(k+2)\alpha$ which is itself upper bounded by $4(N_E(4\alpha)+2)\alpha$ since the edges $E_1, \cdots E_k$ are pairwise distinct. It follows that $d_G(x,y) \leq 4(N_E(4\alpha)+2)\alpha$.
\end{proof}

\begin{theorem} \label{diameter-connected-comp-reeb}
Let $(G, d_G)$ be a connected finite metric graph and let $(X,d_X)$ be a compact geodesic metric space such that $d_{GH}(X,G) < \e$ for some $\e >0$. Let $x_0 \in X$ be a fixed point and let $d_{x_0} = d_X(x_0,.) : X \to [0,+\infty)$ be the distance function to $x_0$.
Then for $d \geq \alpha \geq 0$ the diameter of any connected component $L$ of $d_{x_0}^{-1}([d-\alpha,d+\alpha])$ satisfies
\[
diam(L) \leq 4 (2+N_E(4 (\alpha+2\e))) (\alpha+2\e) + \e
\]
where $N_E(4 (\alpha+2\e))$ is the number of edges of $G$ of length at most $4 (\alpha+2\e)$. 
In particular if $\alpha = 0$ and $8\e$ is smaller that the length of the shortest edge of $G$ then the diameter of $L$ is upper bounded by $17\e$.
\end{theorem}


\begin{proof}
Let $\e' >0$ be such that $d_{GH}(X,G) < \e' < \e$.
Let $C \subset X \times G$ be an $\e'$-correspondence between $X$ and $G$ and $(x_0,r) \in C$. we denote by $d_r = d_G(r,.) : G \to [0,+\infty)$ the distance function to $r$ in $G$. 
Let $x_a, x_b \in L$ and let $(x_a,y_a), (x_b,y_b) \in C$. There exists a continuous path $\gamma \subseteq L$ joining $x_a$ to $x_b$.
Since $C$ is an $\e'$-correspondence for any $x \in \gamma$ there exists a point $(x,y) \in C$ such that $d - \alpha -\e' \leq d_r(y) \leq d+ \alpha + \e'$. The set of points $y$ obtained in this way is not necessarily a continuous path from $y_a$ to $y_b$. However one can consider a finite sequence $x_1 = x_a, x_2, \cdots, x_n = x_b$ of points in $\gamma$ such that for any $i=1,\cdots n-1$ we have $d_X(x_i,x_{i+1}) < \e -\e'$. If $(x_i,y_i) \in C$ then we have $d_G(y_i, y_{i+1}) < \e - \e' + \e' = \e$. As a consequence, since $d- \alpha -\e < d- \alpha - \e' < d_r(y_i) < d + \alpha + \e' < d+ \alpha + \e$ the shortest geodesic connecting $y_i$ to $y_{i+1}$ in $G$ remains in the set $d_r^{-1}([d-\alpha-2\e,d+\alpha+2\e])$ and connecting these geodesics for all $i=1,\cdots,n-1$ we get a continuous path from $y_a$ to $y_b$ in $d_r^{-1}([d-\alpha-2\e,d+\alpha+2\e])$. 
It then follows from Proposition \ref{proposition:band-diameter} that $d_G(y_a,y_b) \leq \leq 4 (2+N_E(4 (\alpha+2\e))) (\alpha+2\e)$ and since $C$ is an $\e'$-correspondence (and so an $\e$-correspondence), $d_X(x_a,x_b) < 4 (2+N_E(4 (\alpha+2\e))) (\alpha+2\e) + \e$.

\end{proof}

As a corollary of the Theorem \ref{diameter-connected-comp-reeb} and Theorem \ref{thm:bound-dGH-betti1-M} we immediately obtain the following results for the Reeb graph and the $\alpha$-Reeb graphs.

\begin{theorem} \label{thm:bound-dGH-final}
Let $(X, d_X)$ be a compact connected path metric space, let $r \in X$ be a fixed base point such that the metric Reeb graph $(G,d_G)$ of the function $d = d_X(r,.): X \rightarrow \mathbb{R}$ is a finite graph.
If for a given $\e >0$ there exists a finite metric graph $(G',d_{G'})$ such that $d_{GH}(X,G') < \e$ then we have 
\[
d_{GH}(X, G) < (\beta_1(G) + 1)  (17 + 8 N_{E,G'}(8 \e)) \e
\] 
where $N_{E,G'}(8\e)$ is the number of edges of $G'$ of length at most $8\e$. 
In particular if $X$ is at distance less than $\e$ from a metric graph with shortest edge length larger than $8\e$ then $d_{GH}(X,G) < 17  (\beta_1(G) + 1) \e$.
\end{theorem}

\begin{theorem} \label{thm:bound-dGH-final-alpha}
Let $(X, d_X)$ be a compact connected path metric space. Let $r \in X$, $\alpha >0$ and $\mathcal{I}$ be a finite covering of the segment $[0,\text{Diam}(X)]$ by open intervals of length at most $\alpha$ such that the $\alpha$-Reeb graph $G_\alpha$ associated to $\mathcal{I}$ and the function $d = d_X(r,.): X \rightarrow \mathbb{R}$ is a finite graph.
If for a given $\e >0$ there exists a finite metric graph $(G',d_{G'})$ such that $d_{GH}(X,G') < \e$ then we have 
\[
d_{GH}(X, G_\alpha) < (\beta_1(G_\alpha) + 1)  ( 4 (2+N_{E,G'}(4 (\alpha+2\e))) (\alpha+2\e) + \e )
\] 
where $N_{E,G'}(4 (\alpha+2\e))$ is the number of edges of $G'$ of length at most $4 (\alpha+2\e)$. 
In particular if $X$ is at distance less than $\e$ from a metric graph with shortest edge length larger than $4 (\alpha+2\e)$ then $d_{GH}(X,G_\alpha) < (\beta_1(G_\alpha) + 1)(8\alpha+17\e)$.
\end{theorem}

\section{Algorithm}  \label{sec:algo-pl}
In this section, we describe an algorithm for computing $\alpha$-Reeb graph for some $\alpha > 0$. 
We assume the input of the algorithm is a neighboring graph $H=(V, E)$, a function  
$l: E \rightarrow \mathbb{R}^+$ specifying the edge length and a parameter $\alpha$. 
In the applications where the input is given
as a set of points together with pairwise distances, i.e., a finite metric space, 
one can generate the neighboring graph $H$ as a Rips graph of the input points with 
the parameter chosen as a fraction of $\alpha$.  

Our algorithm can be described as follows. We assume $H$ is connected 
as one can apply the algorithm to each connected component if $H$ is not . 
Figure~\ref{fig:algorithm} illustrates the different steps of the algorithm.  
In the first step, we fix a node of $H$ as the root $r$ and then obtain the distance function $d: V \rightarrow \mathbb{R}^+$ by
computing $d(v)$ as the graph distance from the node $v$ to $r$. In the second step, 
we apply the Mapper algorithm~\cite{smc-tmah-07} to the nodes $V$ with filter $d$ to construct
a graph $\tilde{G}$. 
Specifically, let $\mathcal{I} = \{(i\alpha, (i+1)\alpha), ((i+0.5)\alpha, (i+1.5)\alpha)| 0\leq i \leq m \}$
so that $\cup_{I_k \in \mathcal{I}} I_k$ covers the range of the function $d$.
We say an interval $I_{k_1}\in \mathcal{I}$ is lower than another interval $I_{k_2}\in \mathcal{I}$ if
the midpoint of $I_{k_1}$ is smaller than that of $I_{k_2}$. 
Now let $H_k$ be the subgraph of $H$ restricted to $V_k = d^{-1}(I_k)$. Namely
two nodes in $H_k$ are connected with an edge if they are in $H$. Notice that 
each subgraph $H_k$ may have several connected components, which can be listed 
in an arbitrary order. Denote $H_k^l$ the $l$-th connected component of $H_k$ and 
$V_k^l$ its set of nodes. Think of $\{V_k^l\}_{k, l}$ as a cover of $V$. Then the graph $\tilde{G}$ 
constructed by the Mapper algorithm is the $1$-skeleton of the nerve of that cover. Namely, each node in $\tilde{G}$ represents
an element in $\{V_k^l\}_{k, l}$, i.e., a subset of nodes in $V$. Two nodes $V_{k_1}^{l_1}$
and $V_{k_2}^{l_2}$ are connected with an edge if $V_{k_1}^{l_1} \cap V_{k_2}^{l_2} \neq \emptyset $. 

In the final step, we represent each node $V_{k}^{l}$ in $\tilde{G}$ using a copy of 
the interval $I_k$. As mentioned in the Section~\ref{sec:reeb-graph}, $\alpha$-Reeb graph is a quotient space of the
disjoint union of those copies of intervals. Specifically, for an edge in $\tilde{G}$, 
let $V_{k_1}^{l_1}$ and $V_{k_2}^{l_2}$ be its endpoints. Then $I_{k_1}$ and $I_{k_2}$ must 
be partially overlapped. We identify the overlap part of these two intervals. After
identifying the overlapped intervals for all edges in $\tilde{G}$, the resulting quotient space 
is the $\alpha$-Reeb graph. Algorithmically, the identification is performed as follows. We split each copy of internal 
$I_k$ into two by adding a point in the middle. Now think of it as a graph with two edges and label one of them upper and 
the other lower.  Notice that two overlapped intervals $I_{k_1}$ and $I_{k_2}$ can not be exactly the same. 
One must be lower than the other. To identify their overlapped part, we identify the upper edge 
of the lower interval with the lower edge of the upper interval.

The time complexity of the above algorithm is dominated by the computation of the distance
function in the first step, which is $O(|E| + |V|\log|V|)$. The computation of 
the connected components in the second step is $O(|V|\log|V|)$ based on union-find data structure.
In the final step, there are at most $O(|V|)$ number of the copies of the intervals. Based 
on union-find data structure, the identification can also be performed in  $O(|V|\log|V|)$ time. 

\begin{figure}[h]
\begin{center}
\begin{tabular}{c}
\includegraphics[width=0.8\textwidth]{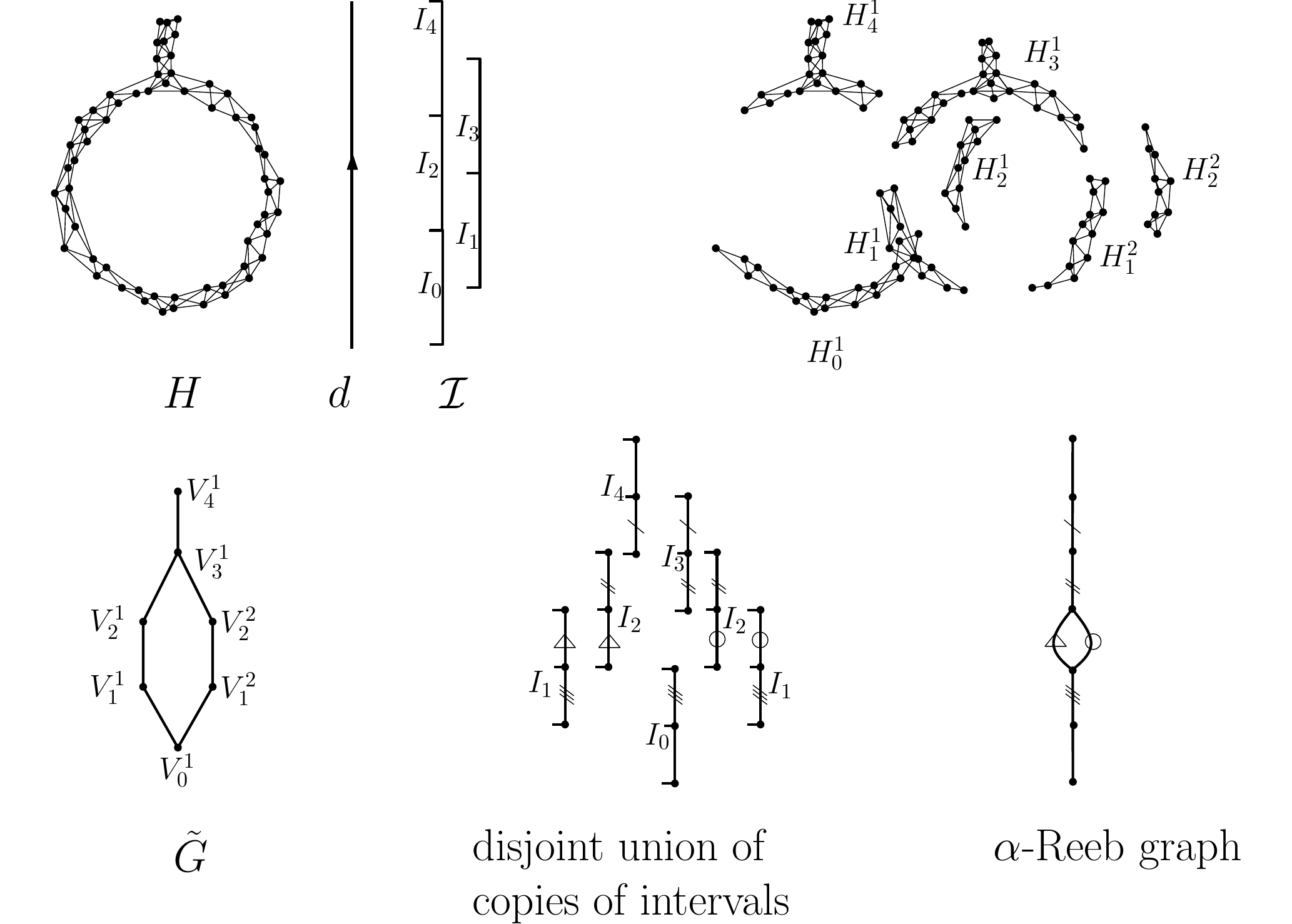} 
\end{tabular}
\caption{
Illustration of the different steps of the algorithm for computing $\alpha$-Reeb graph. 
In the disjoint union of copies of intervals, the subintervals marked with same labels are identified in the $\alpha$-Reeb graph. }
\label{fig:algorithm}
\end{center}
\end{figure}


\section{Experiments}  \label{sec:exp}
In this section, we apply our algorithm to a few data sets. 
The first data set is that of earthquake locations through which we
wish to learn the geometric information about earthquake faults. 
The raw data was obtained from USGS Earthquake Search~\cite{earthquake} and 
consists of earthquakes between 01/01/1970 and 01/01/2010, of magnitude 
greater than $5.0$, and of location in the rectangular area between 
latitudes -75 degrees and 75 degrees and longitude between -170 degrees
and 10 degrees. The raw
earthquake data set contains the coordinates of the epicenters
of 12790 earthquakes. We follow the procedure described in ~\cite{MRIDUL-metric-graph} to 
remove outliers and randomly sampled 1600 landmarks. Finally, we computed 
a neighboring graph from these landmarks with parameter $4$. The 
length of an edge in this graph is the Euclidean distance between its endpoints. 
For each connected component, we fix a root point and compute the graph 
distance function $d$ to the root point as shown in Figure~\ref{fig:earthquake}(a).
Set $\alpha$  also equals $4$ and apply our algorithm to the above data to obtain 
the $\alpha$-Reeb graph. 
In general $\alpha$-Reeb graph is an abstract metric graph. 
In this example, for the purpose of visualization, we use the coordinates 
of the landmarks to embed the graph into the plane as follows. Recall that 
for a copy of interval $I_k$ representing the node $V_k^l$ in $\tilde{G}$, 
we split it into two by adding a point in the middle. We embed the 
endpoints of the interval to the landmarks of the minimum and the maximum 
of the funciton $d$ in $V_k^l$, and the point in 
the middle to the landmark of the median of the function $d$ in $V_k^l$. 
Figure~\ref{fig:earthquake}(b) shows the embedding of the $\alpha$-Reeb graph. 
Note this embedding may introduce metric distortion, i.e., the Euclidean
length of the edge may not reflect the length of the corresponding edge in the $\alpha$-Reeb graph. 

\begin{figure}[h]
\begin{center}
\begin{tabular}{cc}
\includegraphics[width=0.45\textwidth]{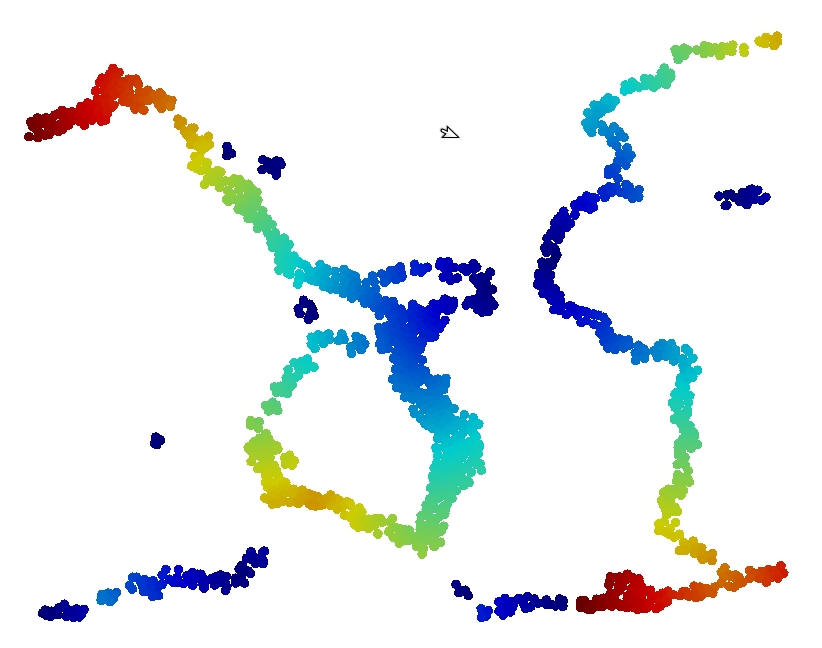} & \includegraphics[width=0.45\textwidth]{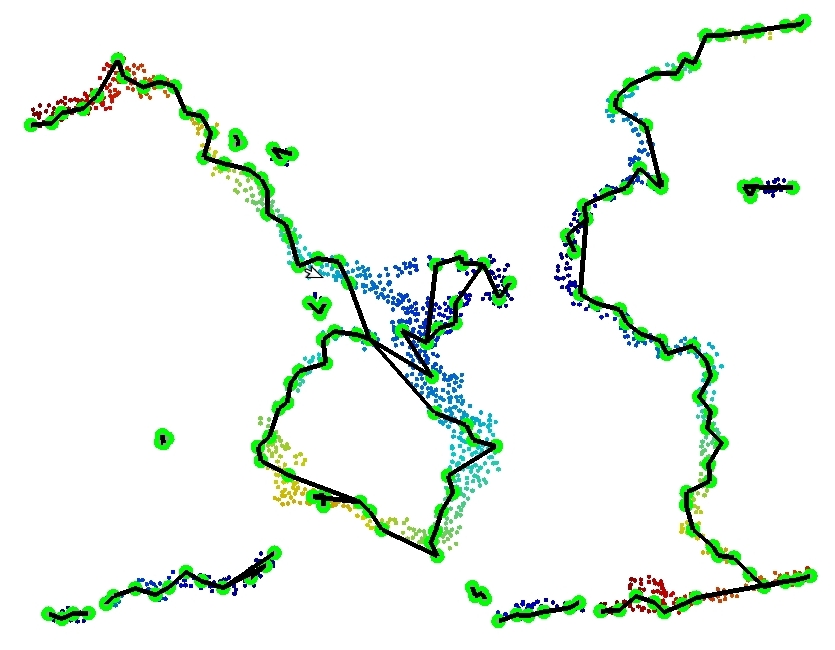}\\
(a) & (b)
\end{tabular}
\caption{(a) The distance functions $d$ on each connected components. The value increases from cold to warm colors. 
(b) The reconstructed $\alpha$-Reeb graph. }
\label{fig:earthquake}
\end{center}
\end{figure}

The second data set is that of 500 GPS traces tagged ``Moscow'' 
from OpenStreetMap~\cite{streetmap}. Since cars move on roads, we expect the
locations of cars to provide information about the metric
graph structure of the Moscow road network. We first selected a metric 
$\epsilon$-net on the raw GPS locations
with $\epsilon = 0.0001$ using furthest point sampling. Then, we
computed a neighboring graph from the samples with parameter $0.0004$.
Again for each connected component, we fix a root point and compute the graph 
distance function $d$ to the root point as shown in Figure~\ref{fig:gps-moscow}(a).
Set $\alpha$ also equals $0.0004$ and compute the $\alpha$-Reeb graph. 
Again, we use the same method as above to embed the $\alpha$-Reeb graph into
the plane, as shown in Figure~\ref{fig:gps-moscow}(b). 

\begin{figure}[h]
\begin{center}
\begin{tabular}{cc}
\includegraphics[width=0.45\textwidth]{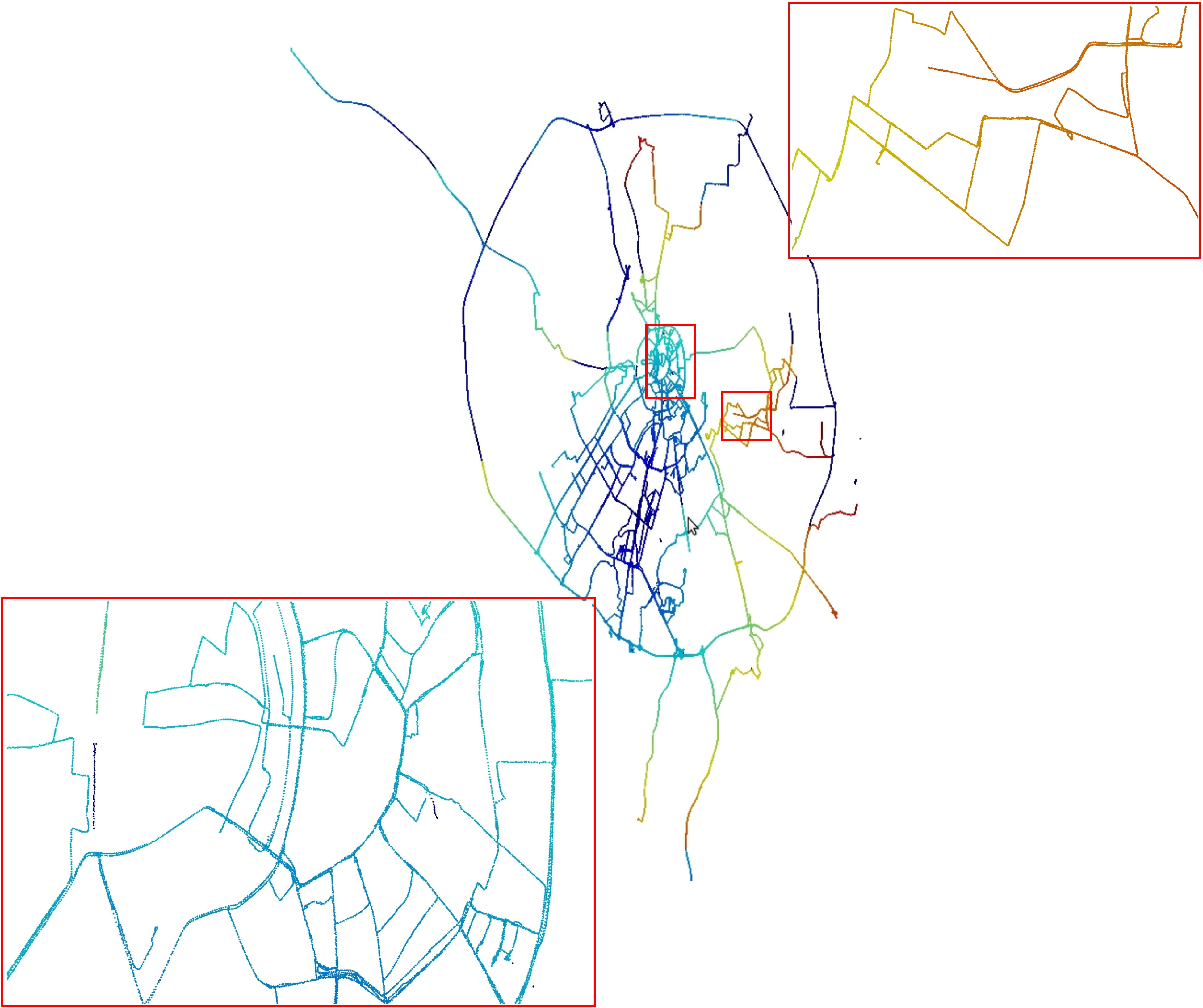} & \includegraphics[width=0.45\textwidth]{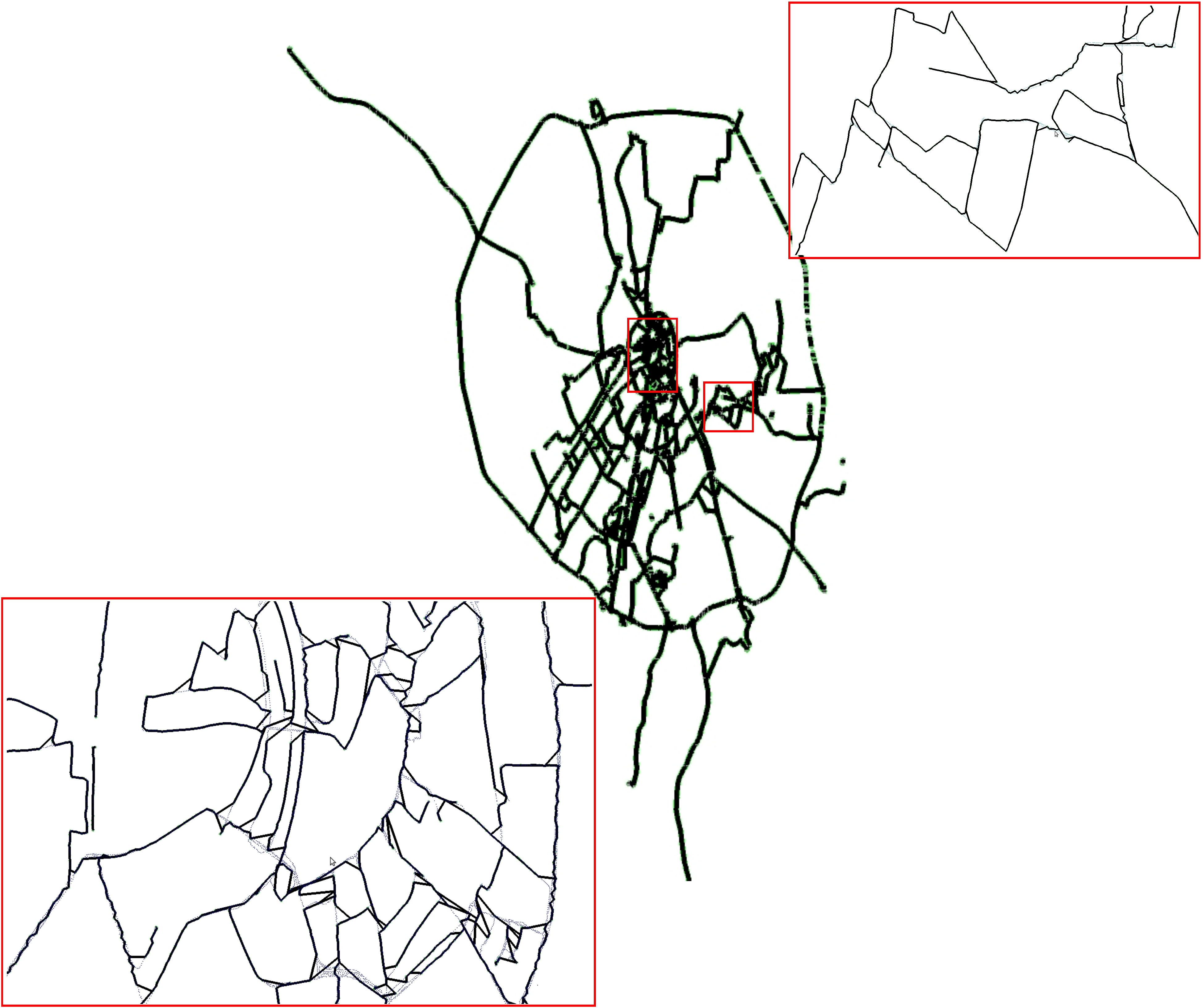} \\
(a) & (b)
\end{tabular}
\caption{(a) The distance functions $d$ on each connected components. The value increases from cold to warm colors.
(b) The reconstructed $\alpha$-Reeb graph. }
\label{fig:gps-moscow}
\end{center}
\end{figure}

To evaluate the quality of our $\alpha$-Reeb graph for each data set, 
we computed both original pairwise distances, and pairwise 
distances approximated from the constructed $\alpha$-Reeb graph. 
For GPS traces, we randomly select 100 points as the data set is too big
to compute all pairwise distances. We also evaluated the use of 
$\alpha$-Reeb graph to speed up distance computations by showing reductions 
in computation time. Only pairs of points in the same connected
component are included because we obtain zero error
for the pairs of vertices that are not. 
Statistics for the size of the reconstructed graph, error of
approximate distances, and reduction in computation time
are given in Table~\ref{tbl:statistics}. 

\begin{table*}[!h]
\begin{center}
\begin{tabular}{ ccc }
\hline
  &   GPS traces  & Earthquake \\
\hline
\#Original points       & 82541 & 1600 \\
\#Original edges        &  313415 & 26996 \\
\#Nodes in $\alpha$-Reeb graph     & 21644 &  147 \\
\#Edges in $\alpha$-Reeb graph     & 21554 &  137 \\
\hline
Graph reconstruction time  & 46.8 & 0.32\\
Original Dist Comp Time & 15.27 & 1.12 \\
Approx Dist Comp Time & 0.96 & 0.01 \\
\hline
Mean distortion       & 6.5\% & 14.1\%\\
Median distortion    & 5.3\% & 12.5\%\\
 \hline
\end{tabular}
\end{center}
\vspace*{-0.15in}
\caption{The graph reconstruction time is the total time
of computing distance function and reconstructing the graph. 
The original distance computation time shows the total
time of computing these distances using the original graph. 
The approximate distance computation time is the total time to compute approximate 
distances based on the reconstructed $\alpha$-Reeb graph. All times are in seconds. \label{tbl:statistics}}
\end{table*}

Road network is directional. There are one-way streets. 
In fact, roads are often split so that cars in different directions
run in different lanes. In particular, this is the true for highways. 
In addition, when two roads cross in GPS coordinates, they may bypass
through a tunnel or an evaluated bridge and thus the road network itself
may not cross. 
Since in most circumstances, drivers follow the road network and 
do not drive against the traffic, such directional information is contained 
in the GPS traces. Here we encode the directional information by stacking several 
consecutive GPS coordinates to form a point in a higher dimensional space. 
In this way, we obtain a new set of points in this higher dimension space.  
Then we build a neighboring graph for this new set of points 
based on $L_2$ norm and apply our algorithm to recover the road network. 
In particular, although the paths intersect at the cross in GPS cooridnates, 
the roadnet work does not. 

We first test the above strategy on a synthetic dataset, 
where we simulate a car driving on a highway crossing according to the right-driving rule
as shown in Figure~\ref{fig:cross-synthetic}(a). 300 hundred traces are obtained, each of which 
contains 500 points. Stack 10 consecutive points along a trace
to form a point in $20$-dimensional space. In this way, obtain a point set in 
$20$-dimensional space. Build a neighboring graph over this set of points
where the length of an edge is the Euclidean distance of its endpoints, and 
compute the distance function as the graph distance to an arbitrarily chosen point 
as shown in Figure~\ref{fig:cross-synthetic}(b) where the points are projected to
three space using the 1st, 2nd and 17th coordinates. To visualize the reconstructed
$\alpha$-Reeb graph, we choose the landmarks in the same way as the previous examples 
and then project them using the above three coordinates. Figure~\ref{fig:cross-synthetic}(c) and 
(d) show the reconstructed graph in two viewpoints. As we can see, we recover the road network.

\begin{figure}[ht]
\begin{center}
\begin{tabular}{cc}
\includegraphics[width=0.4\textwidth]{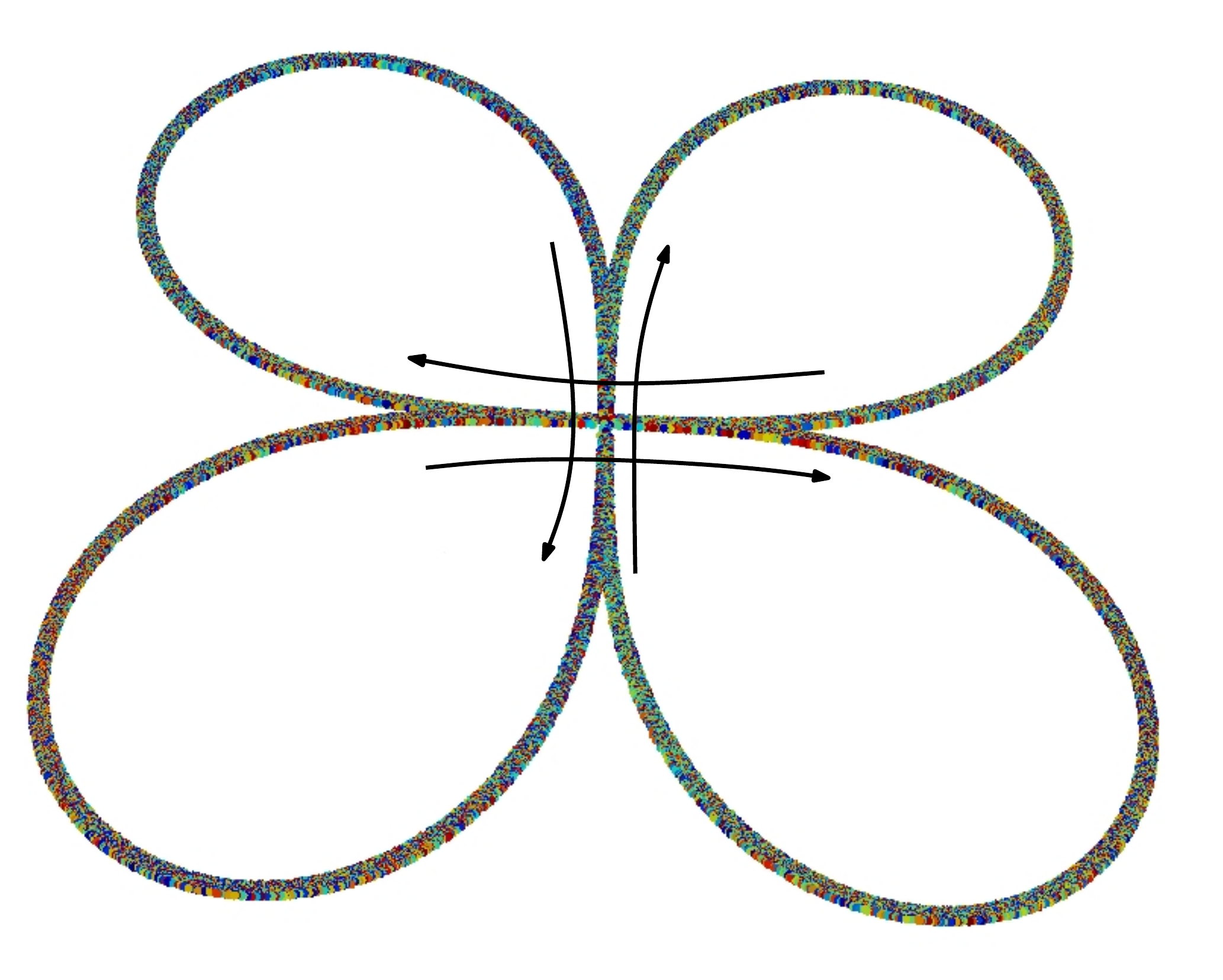} & 
\includegraphics[width=0.4\textwidth]{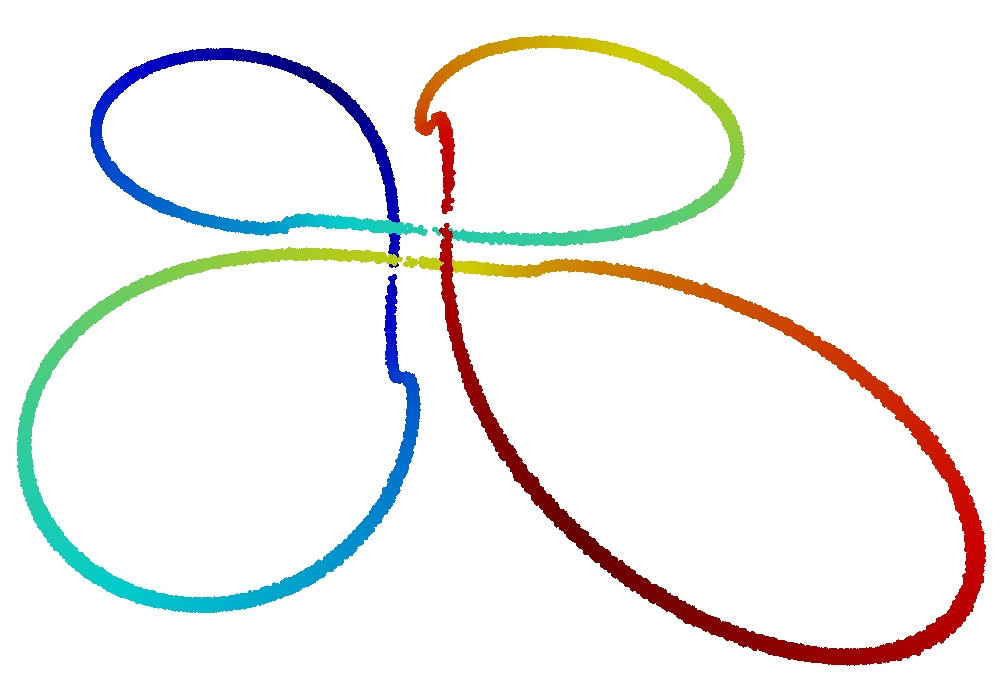} \\
(a) & (b) \\
\includegraphics[width=0.4\textwidth]{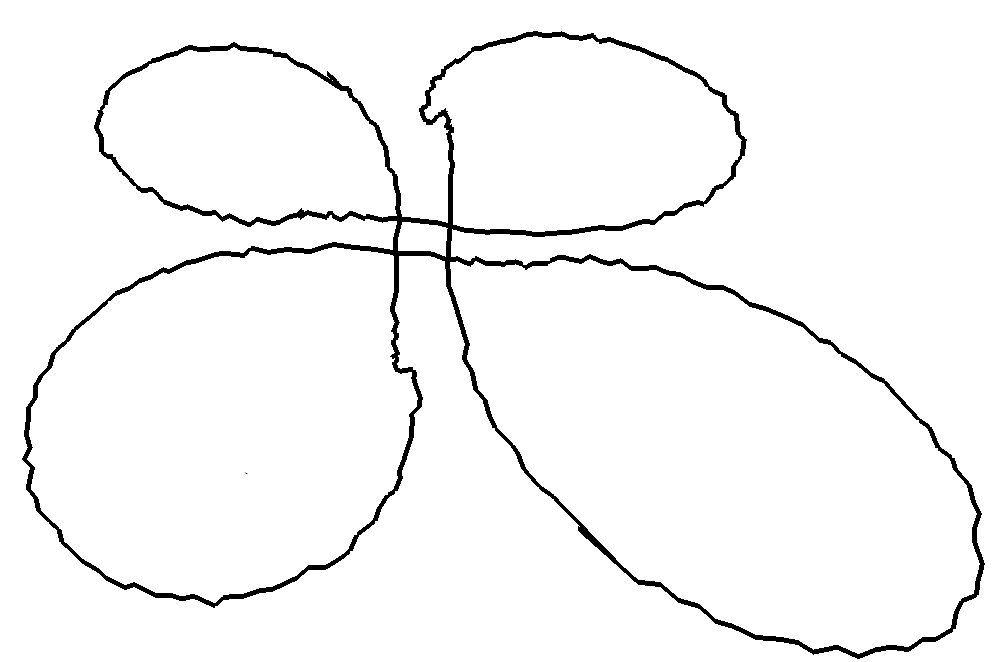}&
\includegraphics[width=0.4\textwidth]{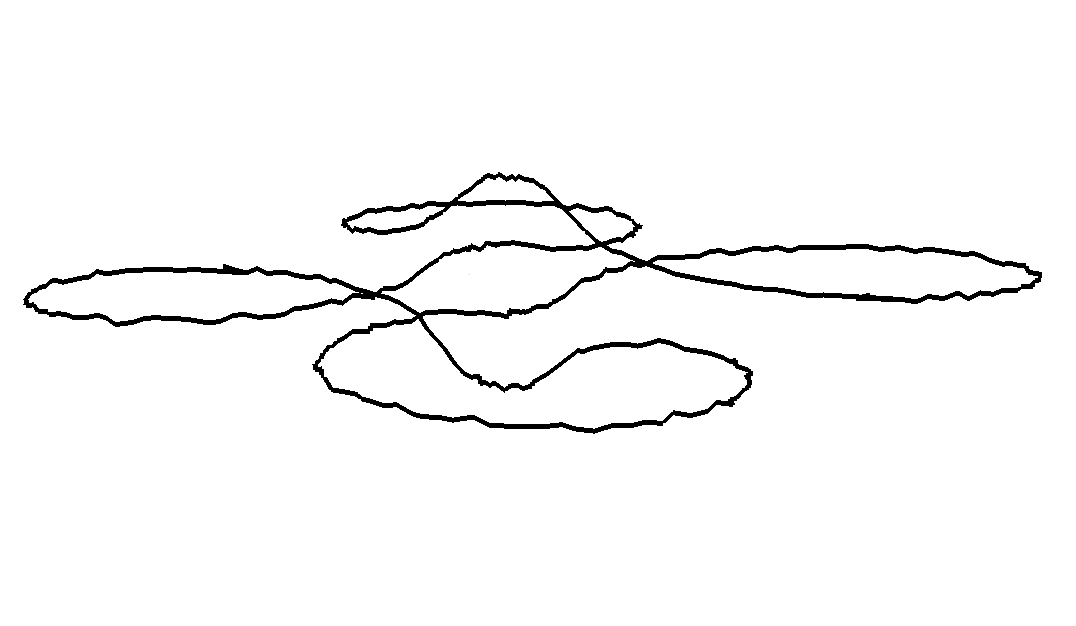}\\
(c) & (d) 
\end{tabular}
\caption{(a) A highway crossing and the synthetic traces. (b) The distance function. 
(c) and (d) The reconstructed $\alpha$-Reeb graph viewed from two perspectives. }
\label{fig:cross-synthetic}
\end{center}
\end{figure}

Next we extract those GPS traces from the above ``Moscow'' dataset which pass through 
a highway crossing as shown in Figure~\ref{fig:cross-moscow}(a). Since GPS records the position based
on time, we resample the traces so that the distances between any two consecutive 
samples is the same among all traces. Then we apply the above algorithm to
the resampled traces. Figure~\ref{fig:cross-moscow}(c) and (d) show the reconstructed graph
which recovers the road network of this highway crossing.

\begin{figure}[ht]
\begin{center}
\begin{tabular}{cc}
\includegraphics[width=0.4\textwidth]{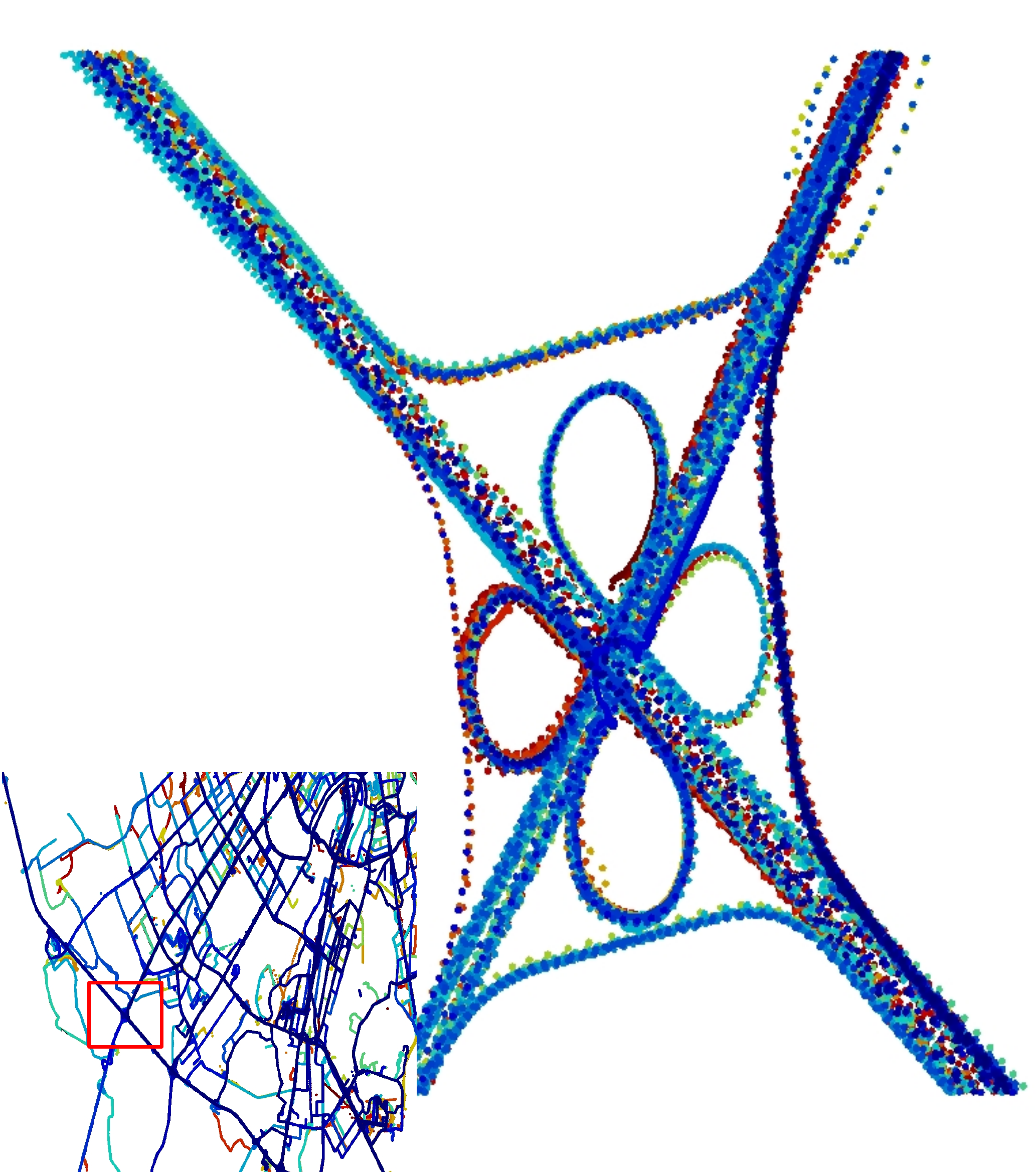} &
\includegraphics[width=0.4\textwidth]{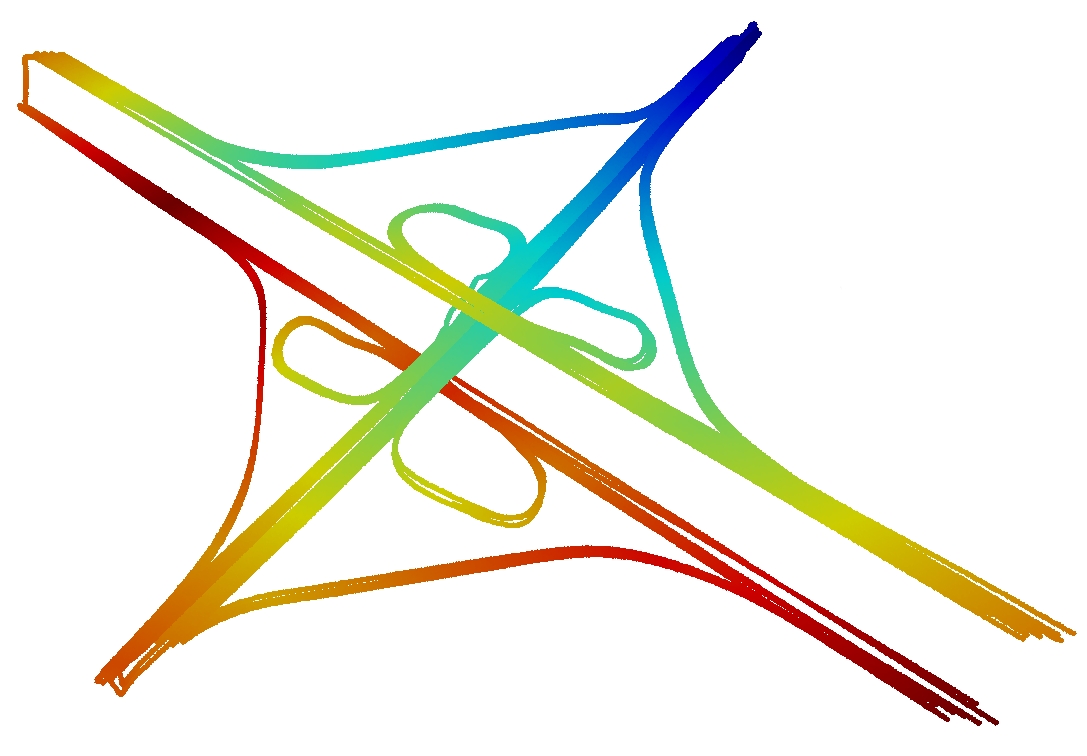} \\
(a) & (b)\\
\includegraphics[width=0.4\textwidth]{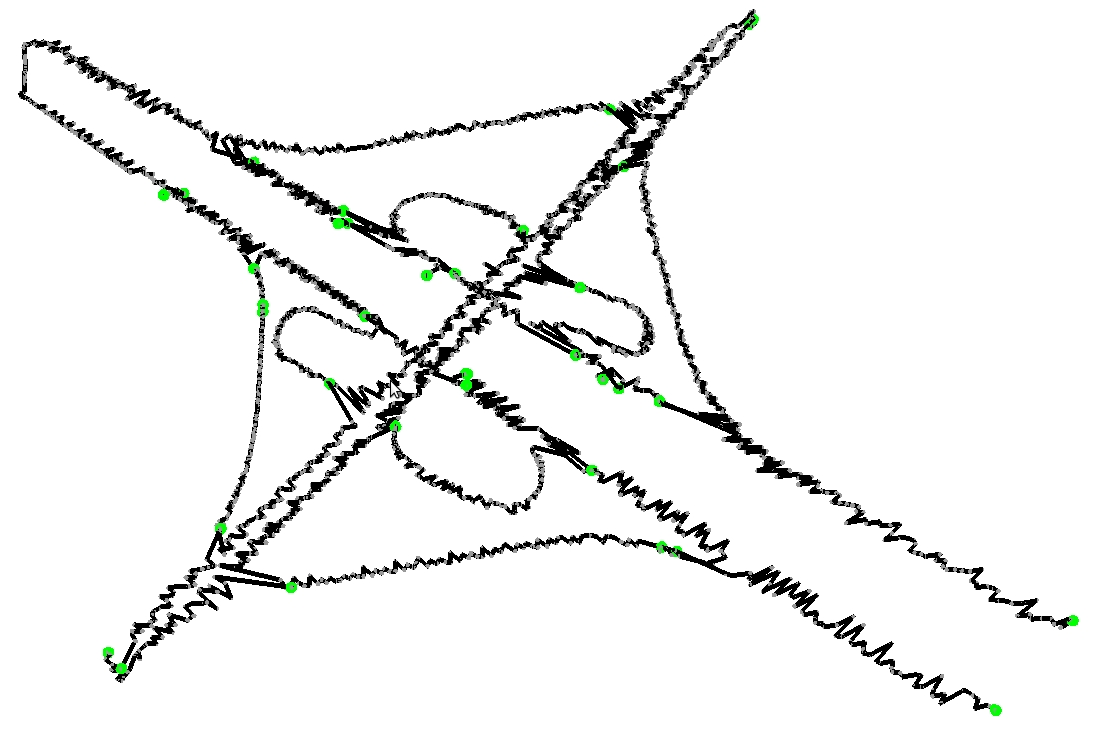}&
\includegraphics[width=0.4\textwidth]{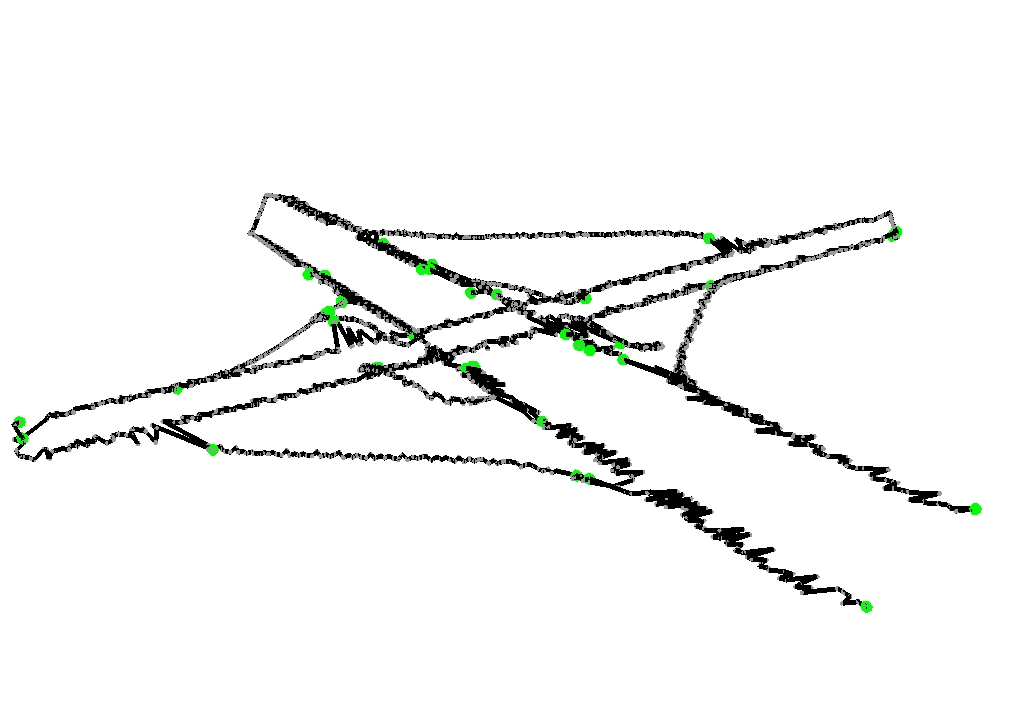}\\
(c) & (d)
\end{tabular}
\caption{(a) GPS traces passing through a highway crossing in Moscow . (b) The distance function.
(c) and (d)The reconstructed $\alpha$-Reeb graph viewed from two perspectives. }
\label{fig:cross-moscow}
\end{center}
\end{figure}

\section{Discussion}  \label{sec:discussion}
We have proposed a method to approximate path metric
spaces using metric graphs with bounded Gromov-Hausdorff distortion, 
and illustrated the performances of our method on a few synthetic 
and real data sets. Here we point out a few possible directions for
future work. First, notice that the $\alpha$-Reeb graph is a quotient space where
the quotient map is 1-Lipschitz and thus the metric only gets contracted. 
In addition, the distance from a point to the chosen root is exactly preserved. 
Therefore, one always reduces the metric distortion by taking the maximum of 
the graph metrics of different root points. It is interesting to study 
the strategy of sampling root points to obtain the smallest metric 
distortion with the fixed number of root points. Second, our method 
in the current form does not recover the topology of the underlying 
metric space. The tools recently developed in persistence homology
seem useful for recovering topology: we provide a preliminary persistence-based result in the Appendix showing that the first Betti number of the underlying metric graph can be inferred from the data. On another hand, Reeb graphs have recently been used for topological inference purpose in \cite{dw-rgap-11}. It would be interesting to combine  our method to these approaches to also obtained topologically correct reconstruction algorithms.
Finally, our method is sensitive to the noise. One can of course
preprocess the data and remove the noise and then apply our algorithm. 
Nevertheless, it is interesting to see if the algorithm can be 
improved to handle noise.


\section*{Acknowledgments}
The authors acknowledge Daniel M\"ullner and G. Carlsson for fruitful discussions and for providing code for the Mapper algorithm. 
They acknowledge the European project CG-Learning EC contract No.~255827; the ANR project GIGA (ANR-09-BLAN-0331-01);
The National Basic Research Program of China (973 Program 2012CB825501); Tsinghua National Laboratory for Information Science 
and Technology（TNList）Cross-discipline Foundation.

\bibliographystyle{plain}
{\small
\bibliography{refPaperGraph}
}

\section*{Appendix}  \label{sec:appendix}
\subsection*{Getting the first Betti number of a graph from an approximation} \label{sec:betti-number}
Although our metric graph reconstruction algorithm does not provide topological guarantees, we show below that, using persistent topology arguments, that the first Betti number of a graph can be inferred from an approximation. 

Recall that given a compact metric space $(X,d_X)$ and a real parameter $\alpha \geq 0$, the Vietoris-Rips complex $\rips(X,\alpha)$ is the simplicial complex with vertex set $X$ and whose simplices are the finite subsets of $X$ with diameter at most $\alpha$:
$$\sigma = [x_0, x_1, \cdots, x_k] \in \rips(X,\alpha) \Leftrightarrow d_X(x_i,x_j) \leq \alpha \ \ {\rm for all } \ i,j.$$

\begin{lemma} \label{lemma:betti1-graph}
Let $G$ be a connected metric graph and let $l(G)$ be the length of the shortest loop in $G$ that is not homologous to $0$. 
For any metric space $D$ such that $d_{GH}(G,D) < \frac{1}{16} l(G)$ and any $d_{GH}(G,D) < \alpha < \frac{3}{16} l(G)$, the first Betti number of $G$ is given by
$$b_1(G) = \mbox{\rm rank}\left ( H_1(\rips(D,\alpha)) \to H_1(\rips(D,3\alpha) \right)$$
where the homomorphism between the homology groups is the one induced by the inclusion maps between the Rips complexes. 
\end{lemma}

\begin{proof}
The proof follows from a result of \cite{h-ovcct-95} that relates the homology of the Rips complexes built on top of $G$ to the homology of $G$ and a result of \cite{cdso-psvrc-12} that allows to relate the Rips filtration built on top of $G$ and $D$ at the homology level. 
Since $G$ is a geodesic path, it follows from Theorem 3.5 and Remark 2), p.179 in \cite{h-ovcct-95} that for any $\alpha < \frac{1}{4} l(G)$, $\rips(G,\alpha)$ and $G$ are homotopy equivalent. Moreover, from Proposition 3.3 in \cite{h-ovcct-95}, for any $\alpha \leq \alpha' < \frac{1}{4} l(G)$, the homomorphism $H_1(\rips(G,\alpha)) \to H_1(\rips(G,\alpha'))$ induced by the inclusion map is an isomorphism. 

Now let $C \subset D \times G$ be an $\varepsilon$-correspondence between $D$ and $G$ where $\varepsilon < \frac{1}{16} l(G)$. According to \cite{cdso-psvrc-12}, the persistence modules $(H_1(\rips(D,\alpha))_{\alpha \in \R_+}$ and $(H_1(\rips(G,\alpha))_{\alpha \in \R_+}$ are $\varepsilon$-interleaved. Now let $\alpha$ be as in the statement of the lemma and let $\beta>0$ be such that $\beta + \varepsilon < \alpha$. The $\e$-interleaving induces the following sequence of homomorphisms
$$H_1(\rips(G,\beta)) \to H_1(\rips(D,\alpha)) \to H_1(\rips(G,\alpha+\e)) \to H_1(\rips(D,3\alpha)) \to H_1(\rips(G,3\alpha+\e))$$
where the composition of two consecutive homomorphisms is the homomorphism induced by the inclusion map between the corresponding Rips complexes. As a consequence since $3\alpha+\e<\frac{1}{4} l(G)$ the homomorphisms $H_1(\rips(G,\beta)) \to H_1(\rips(G,\alpha+\e))$ and $H_1(\rips(G,\alpha+\e)) \to H_1(\rips(G,3\alpha+\e))$ are isomorphisms of rank $b_1(G)$. It follows that the rank of $H_1(\rips(D,\alpha)) \to H_1(\rips(D,3\alpha))$ is equal to $b_1(G)$.
\end{proof}

\end{document}